\documentclass[twocolumn]{autart}

\usepackage{amsmath,amssymb,amsfonts}
\usepackage{xcolor}
\usepackage{graphicx}
\usepackage{algorithm}
\usepackage{algorithmic}
\usepackage{booktabs}
\usepackage{makecell}
\usepackage{multirow}
\usepackage{textcomp}

\usepackage{stfloats}
\usepackage{placeins}
\usepackage{float}
\usepackage{wrapfig}
\usepackage{needspace}
\usepackage[authoryear,longnamesfirst]{natbib}

\newtheorem{definition}{Definition}
\newtheorem{assumption}{Assumption}
\newtheorem{theorem}{Theorem}
\newtheorem{lemma}{Lemma}

\newtheorem{remark}{Remark}

\newenvironment{proof}{\noindent{\bf Proof.}}{\hfill$\square$\par}

\begin{document}

\begin{frontmatter}

\title{Separation Principle for Event-Triggered Prescribed-Time Consensus Tracking of Nonlinear Multi-Agent Systems under DoS Attacks}

% \thanks[footnoteinfo]{This paper was not presented at any IFAC meeting. Corresponding author: Changyun Wen.}

\author[NTU]{Hongjian Chen}\ead{hongjian001@e.ntu.edu.sg},
\author[UM]{Hefu Ye}\ead{hefuye@um.edu.mo},
\author[NTU]{Changyun Wen}\ead{ecywen@ntu.edu.sg}

\address[NTU]{School of Electrical and Electronic Engineering, 
Nanyang Technological University, Singapore 639798, Singapore}

\address[UM]{Faculty of Science and Technology, 
University of Macau, Macau 999078, China}


\begin{abstract}
Despite the recent development of control theory for multi-agent systems (MASs), the highly desirable separation principle is difficult to establish even for linear MASs, let alone for nonlinear ones that rely solely on output measurements under denial-of-service (DoS) attacks. This paper establishes a separation principle for distributed leader-following control of this class of nonlinear MASs, allowing the observer and the controller to be designed independently. For each agent, two parametric Lyapunov equations (PLEs) are employed to generate two symmetric positive-definite matrices, which respectively support the independent design of the controller gain and the observer gain. To ensure that these two parameters do not affect each other, we adopt a matrix pencil formulation to decouple the relevant coupled terms and exploit time-varying feedback to handle potential impacts arising from nonlinearities. Furthermore, we design a hybrid observer that consists of a local state observer for reconstructing unmeasurable follower states and a distributed leader state observer for estimating the inaccessible leader state. Notably, we find that as long as the nonlinearity of all agents satisfies a linear-growth-type condition and the nonlinear model of the leader is available for followers, the separation principle can be established regardless of the presence of event-triggered control and/or admissible DoS attacks. In our method, the selection of design parameters for each agent is elegantly simple, involving only three parameters: one for the prescribed convergence time $t_f$, and the other two for the controller and the hybrid observer, respectively. Moreover, the latter two parameters can be chosen independently from explicit admissible ranges once the system order is specified. Numerical simulations verify the effectiveness of the proposed method.
\end{abstract}

\end{frontmatter}

\section{Introduction}

Output-feedback control is a fundamental problem in control theory since full-state measurements are often unavailable or prohibitively expensive in practical engineering systems. This problem warrants even greater attention in multi-agent systems (MASs), where large-scale networked interactions impose additional sensing, communication, and computation burdens. Among various output-feedback control methodologies, the separation principle is particularly desirable, since it allows the observer and controller to be designed independently and enables the output-feedback controller to be constructed by replacing the unavailable states in a state-feedback law with their estimates. Owing to this modular feature, separation-principle-based designs have been extensively studied for linear systems (\cite{ghodrat2021separation}), linear MASs (\cite{su2021output}), and have also been extended to the internal model based output regulation problem (\cite{maggiore2005output}).

While the separation principle has been widely studied for classical control settings, its extension to prescribed-time control is considerably more challenging,
since both the observer and the controller involve time-varying gains, and this difficulty is further compounded by the fact that the underlying system is not only nonlinear but also time-varying. For single-agent linear systems, prescribed-time separation principle results have been developed in \cite{holloway2019prescribed} by imposing that the scaling power of the time-varying observer gains to exceed that of the controller gains by at least twice the system order. For MASs,  however, analogous separation principle results are still largely absent. The work in \cite{tran2021finite} addresses prescribed-time control of leaderless MASs using a time transformation method. Nevertheless, this result remains confined to linear systems and does not readily extend to nonlinear or more general multi-agent architectures.

The main obstacles to establishing such a separation principle for nonlinear MASs with a leader-following architecture lie in two aspects. The first challenge stems from the fact that certain followers lack direct access to the leader's information (\cite{liu2025distributed,li2020distributed,liu2026attacks}). Consequently, followers must not only design local observers to estimate their own unmeasurable states but also rely on distributed observers to reconstruct the leader's state using only neighboring information. This necessity introduces additional coupling terms into the closed-loop system, which in turn undermines the independence of parameter selection. Although various types of distributed observers, such as distributed leader-state observers and distributed relative-state observers (\cite{ref3,zhang2022global}), have been developed for leader-following problems, most of them are designed under state-feedback control. Consequently, they are not directly applicable to the output-feedback control for nonlinear MASs, where only output measurements are available and one must simultaneously reconstruct follower states and estimate the leader state.

The second aspect lies in the treatment of nonlinearities in stability analysis. Most existing methods for handling coupled nonlinear terms, either directly or indirectly, alter the control structure or parameter selection, thereby rendering the separation principle inapplicable. For example, robust treatments based on Young's inequality require incorporating additional negative feedback terms into the control law to counteract state-dependent effects (\cite{li2020distributed}). Similarly, small-gain conditions impose additional constraints on parameter selection (\cite{jin2022leader}). A notable exception is the work of \cite{atassi1999separation}, which established a separation principle for a class of nonlinear systems using high-gain observers. By increasing the observer gain $\epsilon$, the trajectories under output feedback can be made close to those under state feedback, thereby ensuring the separation principle. More recently, \cite{ref4} successfully established a separation principle for a class of nonlinear systems via a PLE-based output feedback controller, revealing that adding an extra design degree of freedom to the matrix pencil formulation developed in \cite{ref-NYU-SCL,ref7} enables the decoupling of coupled terms in the controller and observer, while the blow-up time-varying gains are capable of dominating nonlinear coupling terms, thereby paving the way for establishing the separation principle.

Motivated by the above observations, this paper aims to design a distributed prescribed-time controller and establish the separation principle for a class of leader-following multi-agent systems. First, a hybrid observer comprising a distributed observer and a local observer is constructed: the former estimates the leader's state for followers that cannot communicate directly with the leader, while the latter reconstructs unmeasurable local states for all agents. Notably, a key to the success of the separation principle lies in the ability of the proposed local observer to access the nonlinear model of the leader.   Second, a matrix pencil formulation with an increased design degree-of-freedom is developed to decouple the coupled terms, ensuring that the design parameters of the controller and the observer do not interfere with each other. The resulting control framework involves only three design parameters: the prescribed convergence time, the observer gain, and the controller gain, where the latter two parameters can be obtained from a numerical table that depends solely on the system order. We further investigate the feasibility of the proposed scheme under event-triggered communication and DoS attacks, which is necessary and important in the study of networked systems (\cite{deng2020dynamic,deng2022resilient,yan2025event,liu2026attacks}). We find that if the triggering condition is designed using only controller parameters, and the  DoS-induced switching terms  are incorporated into the resilient matrix-pencil formulation while a mild constraint is imposed on the choice of the prescribed time, then the separation principle still holds for both scenarios.

The main contributions are summarized as follows:
\begin{enumerate}
    \item[$i)$] A separation principle is established for the proposed output-feedback control framework, enabling the observer and controller to be designed independently. As a result, the parameters of the two modules can be selected separately from explicit admissible ranges, without iterative trial-and-error, even in the presence of intermittent communication and DoS attacks.
    \item[$ii)$]A hybrid observer is proposed to address the concurrent unavailability of leader-state access and follower-state measurements for nonlinear leader-following MASs. By synergizing a distributed observer with a high-gain local observer, it enables each follower to retrieve the requisite feedback information using only output measurements and neighbor information. 
    \item[$iii)$] A resilient matrix-pencil-based analysis is proposed, which integrates the coupled terms arising from DoS-induced switching terms, nonlinearities, and both distributed and local observation errors into several generalized eigenvalue solving problems. Through this formulation, explicit relationships are derived between the admissible DoS duration and the extended prescribed time. As an additional benefit, the conservatism in parameter selection is notably reduced.
    \item[$iv)$] Compared with the result in \cite{ref-Fenggang-Auto}, which considers linear
MASs and adopts a full-state feedback design, the proposed method addresses
resilient output-feedback consensus tracking for high-order nonlinear MASs
with unknown growth rates under intermittent communication interruptions,
thereby offering broader applicability to systems with nonlinear dynamics and
limited state measurements.
\end{enumerate}

\textit{Notations:} Let $\mathbb{R}$, $\mathbb{R}^n$, and $\mathbb{R}^{n\times m}$ denote the sets of real numbers, $n$-dimensional real vectors, and $n\times m$ real matrices, respectively. The symbol $\mathbf{0}$ denotes a zero vector or matrix with compatible dimensions, and $\mathbf{I}_n$ denotes the $n\times n$ identity matrix. For scalars or matrices $q_1,\ldots,q_n$, $\operatorname{diag}\{q_1,\ldots,q_n\}$ denotes the corresponding diagonal or block-diagonal matrix. For a vector $x=[x_1,\ldots,x_n]^\top$, $x^\top$ denotes its transpose, $|x|$ its Euclidean norm, and $|x|_e=[|x_1|,\ldots,|x_n|]^\top$. The notation $|q|_e\le_e |x|_e$ means that the corresponding componentwise inequalities hold. For a square matrix $\mathbf{Q}$, $\det(\mathbf{Q})$ denotes its determinant, $|\mathbf{Q}|_e$ denotes its entrywise absolute value, and $\operatorname{diag}\{\mathbf{Q}\}$ denotes the diagonal matrix formed by the diagonal entries of $\mathbf{Q}$.

\section{Preliminaries}

\subsection{Properties of matrix pencils and PLEs}
A useful lemma is introduced for the matrix pencil formulation here:
\begin{lemma}\label{lemma_matrix pencil}
 (\cite{ref7}) Given real square matrices $\mathbf{Q}_1$ and $\mathbf{Q}_2$, the generalized eigenvalues of the
matrix pencil $\mathbf{Q}_1 - s \mathbf{Q}_2$ are defined as the values of $s$ that make
$\det(\mathbf{Q}_1 - s \mathbf{Q}_2) = 0$. The set of generalized eigenvalues of the matrix pencil
$\mathbf{Q}_1 - s \mathbf{Q}_2$ is denoted as $\sigma(\mathbf{Q}_1, \mathbf{Q}_2)$. Generalized eigenvalues have
several useful almost self-evident properties such as the following:
\begin{enumerate}
\item $s \mathbf{Q}_1 + \mathbf{Q}_2 < 0$ holds for all
$s > \max\{\sigma(\mathbf{Q}_2,-\mathbf{Q}_1)\}$, if $\mathbf{Q}_2$ is symmetric and $\mathbf{Q}_1$ is symmetric
negative-definite (SND);
\item $s \mathbf{Q}_1 - \mathbf{Q}_2 < 0$ holds for all
$s < \min\{\sigma(\mathbf{Q}_2,\mathbf{Q}_1)\}$, if $\mathbf{Q}_2$ and $\mathbf{Q}_1$ are symmetric positive-definite
(SPD).
\end{enumerate}
\end{lemma}

Consider the following PLE:
\begin{align}\label{eq:PLE1}
    \mathbf{A}^{\top} \mathbf{P}_i(r_i) + \mathbf{P}_i(r_i) \mathbf{A} - \mathbf{P}_i(r_i) B B^{\top} \mathbf{P}_i(r_i) = -r_i \mathbf{P}_i(r_i),
\end{align}
and its dual form
\begin{align}\label{eq:PLE2}
    \mathbf{A}\mathbf{Q}_i(\varrho_i) + \mathbf{Q}_i(\varrho_i)\mathbf{A}^{\top} - \mathbf{Q}_i(\varrho_i)C^{\top} C\mathbf{Q}_i(\varrho_i) = -\varrho_i \mathbf{Q}_i(\varrho_i),
\end{align}
where $r_i(t)$ and $\varrho_i(t)$ are time-varying parameters to be designed in \eqref{eq:r} and \eqref{eq:varrho}, $\mathbf{A}$, $B$ and $C$ are defined as
\begin{align}\label{eq:A B C}    
\mathbf{A} =
\begin{bmatrix}
\mathbf{0}_{n-1} & \mathbf{I}_{n-1} \\
0 & \mathbf{0}_{n-1}^{\top}
\end{bmatrix},
\quad
B =
\begin{bmatrix}
\mathbf{0}_{n-1} \\
1
\end{bmatrix},
\quad
C =
\begin{bmatrix}
1 \\
\mathbf{0}_{n-1}
\end{bmatrix}^{\top}.
\end{align}
Two lemmas regarding PLEs with the matrix pencil are introduced here.
\begin{lemma}\label{lemma dP_i}
 (\cite{ref4}) Let $(\mathbf{A},B)$ be given by \eqref{eq:A B C}. Then, the PLE \eqref{eq:PLE1} has a unique symmetric
positive definite (SPD) solution $\mathbf{P}_i(r)$ if and only if $r_i>0$. In this case, the
unique solution satisfies
\begin{equation}
\label{eq:4}
\mathbf{P}_i(r_i) = r_i \mathbf{L}_{ci} \mathbf{P}_{ni} \mathbf{L}_{ci},
\end{equation}
where $\mathbf{P}_{ni} = \mathbf{P}_i(1)$ is a constant SPD matrix and $\mathbf{L}_{ci} = \mathrm{diag}\{ r_i^{n-1}, r_i^{n-2}, \ldots, 1 \}$ is a time-varying diagonal matrix. Furthermore, we have
\begin{equation}
\begin{aligned}
&\frac{d\mathbf{P}_i}{dr_i} 
\leq (1+\xi_i)\frac{\mathbf{P}_i}{r_i}, \\
&\mathbf{P}_iBB^\top \mathbf{P}_i 
\leq \delta_i r_i \mathbf{P}_i,
\end{aligned}
\label{eq:5}
\end{equation}
where $\xi_i>0$ and $\delta_i>0$  are constants obtained by  
\begin{align}
\xi_i
&=\max\big\{\sigma(
\mathbf{E}_i\mathbf{P}_{ni}+\mathbf{P}_{ni}\mathbf{E}_i,
\mathbf{P}_{ni})\big\}, \nonumber\\
\delta_i
&=\max\big\{\sigma(
\mathbf{P}_{ni}\mathbf{L}_{ci}BB^{\top}\mathbf{L}_{ci}\mathbf{P}_{ni},
\mathbf{P}_{ni})\big\}.
\end{align}
with $\mathbf{E}_i = \text{diag}\{n-1, n-2, \ldots, 0\}$.
\end{lemma}

\begin{lemma}\label{Lemma dQ}
(\cite{ref6}) Let $(\mathbf{A},C)$ be given by \eqref{eq:A B C}. Then the PLE \eqref{eq:PLE2} has a unique symmetric
positive definite solution $\mathbf{Q}_i(\varrho_i)$ if and only if $\varrho_i > 0$. In this
case, the unique solution satisfies
\begin{equation}
\mathbf{Q}_i(\varrho_i) = \varrho_i^{2n-1} \mathbf{L}_{oi}^{-1} \mathbf{Q}_{ni} \mathbf{L}_{oi}^{-1},
\label{eq:7}
\end{equation}
where $\mathbf{Q}_{ni} = \mathbf{Q}_i(1)$ and $\mathbf{L}_{oi} = \mathrm{diag}
\{\varrho_i^{\,n-1}, \varrho_i^{\,n-2}, \ldots, 1\}$.
In addition, it holds that
\begin{equation}
\frac{d\mathbf{Q}_i}{d\varrho_i} \geq \frac{\mathbf{Q}_i}{n\varrho_i}.
\label{eq:8}
\end{equation}
\end{lemma}

\subsection{System model}
This paper considers nonlinear MASs consisting of one leader and $N$ followers. The dynamic of the leader is
\begin{align}\label{eq:leader}
    &\Dot{x}_0 = \mathbf{A} x_0 + f_0(t,x_0), \nonumber\\
    &y_0 = C x_0,
\end{align}
where $x_0\in \mathbb{R}^n$ and $y_0\in \mathbb{R}$ denote the state and output, respectively. The system matrices $\mathbf{A}$, $B$ and $C$ are given in \eqref{eq:A B C}. The structure of the nonlinear function of $f_0(\cdot,t)$ is available to all followers. The dynamic of the follower is
\begin{align}\label{eq:system}
    &\Dot{x}_i = \mathbf{A} x_i + B u_i + f_i(t,x_i), \nonumber\\
    &y_i = C x_i, \quad i = 1,\ldots,N,
\end{align}
where $x_i\in \mathbb{R}^n$, $y_i\in \mathbb{R}$, $u_i\in \mathbb{R}$ denote the state, output, and control input, respectively. $f_i=[f_{i,1},\cdots,f_{i,n}]^\top \in \mathbb{R}^n$ denotes the nonlinear
function of follower $i$, which satisfies a Lipschitz condition with respect to
$x_i$ and is piecewise continuous with respect to $t$.
 For the closed-loop MASs, the following linear growth condition should be satisfied.
\begin{assumption}\label{ass:nonlinearity}
(\cite{ref-ye2025distributed}) There exist unknown constants $\theta_i$ and a known lower-triangular matrix $\mathbf{A}_f \in \mathbb{R}^{n \times n}$ such that
\begin{equation}\label{eq:lipschitz_phi}
\lvert f_i(x_i,t)-f_j(x_j,t) \rvert_e
\le_e
\theta_i \mathbf{A}_{f} \lvert x_i-x_j \rvert_e .
\end{equation}
where $\theta_i$ denotes the unknown growth rate, and $\mathbf{A}_{f}$ is
\begin{equation}\label{eq:Af}
\mathbf{A}_{f}
=
\begin{bmatrix}
1 & \cdots & 0 \\
\vdots & \ddots & \vdots \\
1 & \cdots & 1
\end{bmatrix}.
\end{equation}
\end{assumption}
\begin{remark}
Assumption \ref{ass:nonlinearity} is a commonly used growth condition in the
control of nonlinear MASs (\cite{zhang2015leader,li2021leader,ref-ye2025distributed}). Compared with
\cite{zhang2015leader,li2021leader}, this paper allows the nonlinear growth
rates to be unknown. Compared with \cite{ref-ye2025distributed}, where each
follower is assumed to have direct access to the leader information, this paper
considers a more general directed communication graph that only contains a
spanning tree rooted at the leader. 
\end{remark}

\subsection{Graph theory}

 Define $\mathcal{G}$ as a directed graph with a node set
$\mathcal{V} = \{0,1,\ldots,N\}$ and an edge set $\mathcal{E} \subseteq \mathcal{V} \times \mathcal{V}$.
If node $j$ communicates with node $i$ directly, one can get that $(j,i) \in \mathcal{E}$.
Define $\mathcal{A} = [a_{ij}] \in \mathbb{R}^{(N+1)\times(N+1)}$ as the adjacency matrix, where
$a_{ij} > 0$ if $(j,i) \in \mathcal{E}$ and $a_{ij} = 0$ otherwise. For agent $i$, $j$ is in the neighboring set $\mathcal{N}_i$ if $a_{ij} > 0$.
Then the Laplacian matrix is defined as
$
\mathcal{L} = \mathcal{H} - \mathcal{A},
$
where $\mathcal{H} = \mathrm{diag}\{\omega_0,\omega_1,\ldots,\omega_N\}$ and
$\omega_i = \sum_{j=0}^{N} a_{ij}$.
$\mathcal{L}$ can be written as
$ 
\mathcal{L} =
\begin{bmatrix}
\mathbf{0} & 0 \\
\mathcal{L}_1 & \mathcal{L}_2 
\end{bmatrix},
$
with $\mathcal{L}_1 \in \mathbb{R}^{N\times N}$ and $\mathcal{L}_2 \in \mathbb{R}^{N\times 1}$.
The following assumption for the directed graph is widely adopted in consensus tracking problems.
\begin{assumption}\label{Ass:Graph} (\cite{olfati2004consensus})
In the absence of DoS attacks, the directed graph $\mathcal{G}$ contains a directed spanning tree rooted at the leader.
\end{assumption}

The objective of this paper is to develop distributed control schemes for
the MAS \eqref{eq:leader}--\eqref{eq:system} when only the outputs of all agents are available,  such that the prescribed-time consensus tracking 
problem formulated in Definition 1 is solved with a user-prescribed time. In addition, we establish the separation principle for distributed control schemes, namely, the design parameters for the controllers and observers are independent of each other.

\begin{definition}
(\cite{ref-ye2025distributed}) The prescribed-time consensus tracking of the leader-following MAS
\eqref{eq:system} is said to be achieved if, for $i\in\{1,\ldots,N\}$,
there exist a class $\mathcal{KL}$ function $\beta$ and a time-scaling
function $\mu(t):[0,t_f)\to[0,\infty)$ such that $\mu(t)$ tends to
infinity as $t$ goes to $t_f$, and, for any
$x_i(0),x_0(0)\in\mathbb{R}^n$,
\begin{equation}
    |x_i(t)-x_0(t)|
    \le
    \beta\big(|x_i(0)-x_0(0)|,\mu(t)\big),
    \quad \forall t\in[0,t_f),
    \label{eq:prescribed_time_consensus}
\end{equation}
where $t_f>0$ is an arbitrary time-independent constant.
\end{definition}

\subsection{Prescribed-time state-feedback controller}
The prescribed-time state-feedback controller for follower $i$ in \eqref{eq:system} that has direct access to $x_0$, without an event-triggered mechanism is given as follows
\begin{align}\label{eq:u1}
    u_i = -\frac{1}{2} B^{\top}\mathbf{P}_i(r_i)(x_i - x_0),
\end{align}
where $r_i(t) = \gamma_{ci} \alpha(t)$ and the time scaling function is
\begin{equation}
    \alpha(t)=\frac{1}{t_f-t}, \quad t\in[0,t_f),
    \label{eq:alpha}
\end{equation}
which are widely adopted in \cite{orlov2022autonomous,orlov2022time,orlov2024scaling,efimov2026discretization,li2022prescribed,li2023prescribed,krishnamurthy2020dynamic}.
Inspired by the state-feedback control design in \cite{ref4}, a lemma is introduced as follows.
\begin{lemma}\label{lemma_state}
Consider a follower that has direct access to the leader state \(x_0\) in the MAS \eqref{eq:system}. With the state-feedback controller \eqref{eq:u1}, if \(\gamma_{ci}\in\Omega_{ci}\), where
\begin{align} \label{Omega_ic}
		\Omega_{ci}
		=
		\left\{
		\gamma_{ci} \in \mathbb{R}
		\,\middle|\,
		\gamma_{ci} > \max\{4+2\xi_i,\ t_f\}
		\right\},
\end{align}
then the prescribed-time consensus tracking objective is achieved in the sense of Definition~1. In addition, all closed-loop signals remain uniformly bounded.
\end{lemma}
\begin{proof}
The result can be obtained by following the same line of reasoning as in Section~3.1 of \cite{ref4}, with the state variable \(x\) replaced by the tracking error \(x_i-x_0\). 
\end{proof}
\begin{remark}
    In \cite{ref4}, the admissible range is given by
\[
\Omega_{ci}
=
\left\{
\gamma_{ci} \in \mathbb{R}
\,\middle|\,
\gamma_{ci} > 4+2\xi_i
\right\}.
\]
Here, we use a slightly stronger condition \(\gamma_{ci}>\max\{4+2\xi_i,t_f\}\), which ensures \(r_i(t) \ge 1\) on the prescribed-time interval and facilitates the subsequent analysis. When \(t_f<4+2\xi_i\), this condition reduces to the admissible range used in \cite{ref4}; see Table~\ref{tab:parameter_ranges}.
\end{remark}

The above result follows the standard route for establishing a separation principle, as in \cite{holloway2019prescribed,tran2021finite,ref4}: a state-feedback controller is first developed, and the corresponding output-feedback controller is obtained by replacing the unavailable state with its estimate. Consequently, for the nonlinear MASs, a hybrid observer is introduced in this paper to reconstruct the follower state and estimate the leader state simultaneously, so that the tracking error can be generated from two estimated variables and the resulting controller can be implemented by any follower over a directed graph under Assumption \ref{Ass:Graph}.

\section{Prescribed-time event-triggered consensus tracking via a separation principle}
Building on the preliminary results, this section further establishes the separation principle for output-feedback event-triggered consensus tracking of nonlinear MASs in the absence of DoS attacks. Before presenting the observer and controller designs, we first provide the overall control block diagram in Fig.~\ref{fig:control_scheme}.

\begin{figure}[!t]
    \centering
    \includegraphics[width=2.4in]{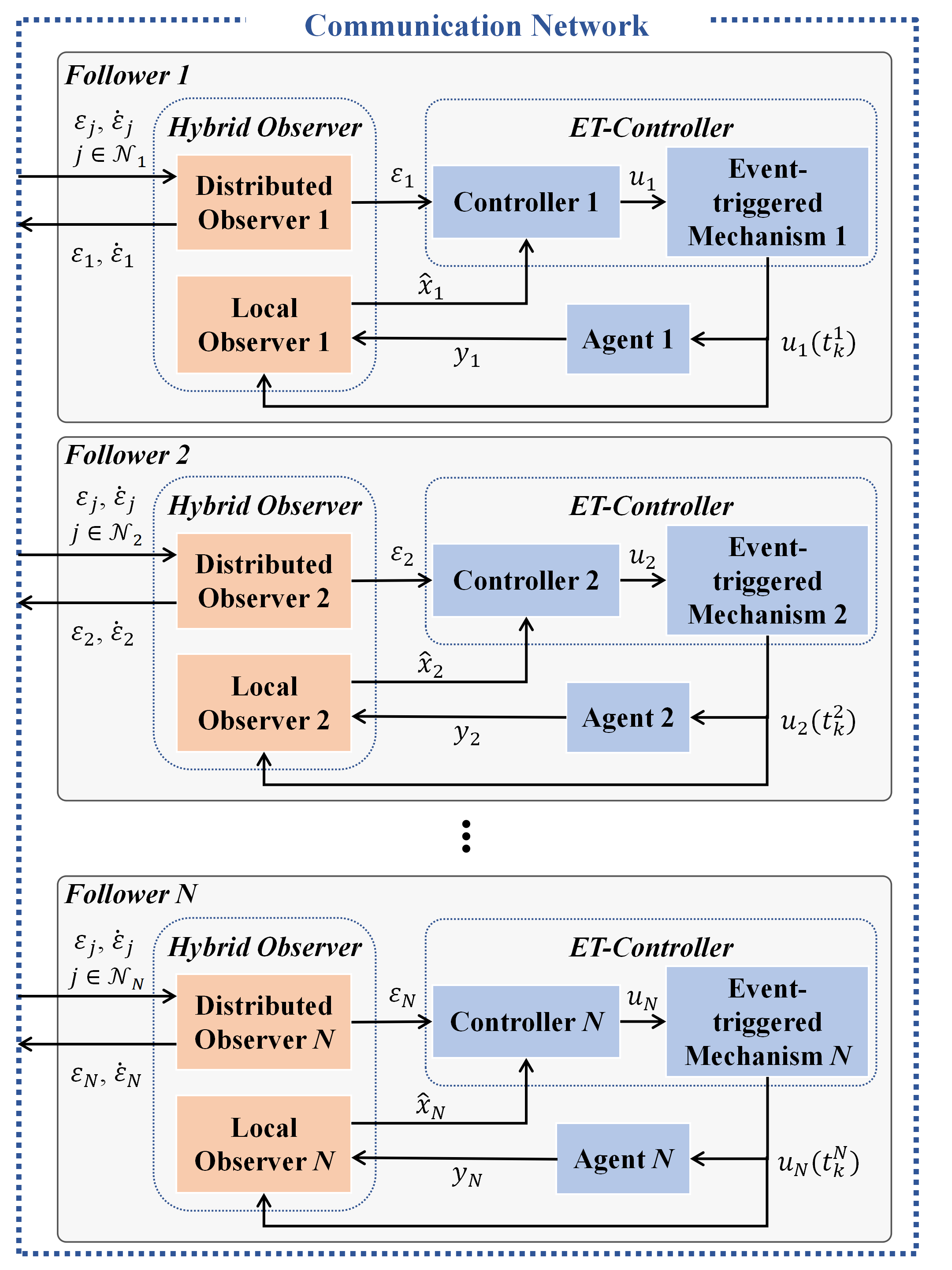}
    \caption{Schematic diagram of the proposed distributed event-triggered output-feedback control scheme without DoS attacks.}
    \label{fig:control_scheme}
\end{figure}

\subsection{Hybrid observer design}
A hybrid observer framework is developed by integrating a distributed observer
with a local observer, where the former estimates the leader state through
local information exchange over the directed graph and the latter reconstructs
the unmeasurable follower state.

\subsubsection{Distributed observer}
Based on Assumption \ref{Ass:Graph}, a subset of agents cannot access the information of the leader directly. To this end, distributed observers are utilized to observe the state of the leader. Inspired by \cite{ref3}, a distributed observer is designed as follows
\begin{align} \label{dis_observer}
\dot{\varepsilon}_i
=
\frac{1}{\sum_{j=0}^{N} a_{ij}}
\sum_{j=0}^{N} a_{ij} \dot{\varepsilon}_j
-
\frac{\varrho_i}{2}
\frac{1}{\sum_{j=0}^{N} a_{ij}}
\sum_{j=0}^{N} a_{ij}
\big( \varepsilon_i - \varepsilon_j \big),
\end{align}
where $\varepsilon_i$ represents agent $i$'s local estimate of the leader's state,   $\varepsilon_0(t) = x_0(t)$, and  $\varrho_i$ is defined in \eqref{eq:varrho}. Let  $\zeta_i(t) = \varepsilon_i(t) - x_0(t)$ be the distributed observation error, and    $\eta_i = \sum_{j=0}^{N} a_{ij}(\varepsilon_i - \varepsilon_j)$ be the local consensus error for agent $i$. Then the stack vectors for followers are given by: $\varepsilon = [\varepsilon_1^\top,\cdots,\varepsilon_N^\top]^\top$, $\zeta = [\zeta_1^\top,\cdots,\zeta_N^\top]^\top$ and $\eta = [\eta_1^\top,\cdots,\eta_N^\top]^\top$. According to Assumption~\ref{Ass:Graph}, we know that $\mathcal{L}_1$ is nonsingular, then  the compact-form of $\zeta$ can be formulated as $\zeta = (\mathcal{L}_1^{-1}\otimes I_n)\,\eta$.
Since all followers have the same system order and share the same prescribed
time $t_f$, the observer design coefficients can be consistently selected from
Table~\ref{tab:parameter_ranges}. Together with the common time-scaling
function $\alpha(t)$, we restrict
$\varrho_1(t)=\cdots=\varrho_N(t)$ for this paper. Therefore, it follows from
\eqref{dis_observer} and the definition of $\eta$ that $\dot{\eta}=-\frac{\varrho_i}{2}\eta$, which further yields
\begin{align}\label{eq:dev zeta}
    \dot{\zeta}_i=-\frac{\varrho_i}{2}\zeta_i.
\end{align}

\subsubsection{Local observer}
For system \eqref{eq:system} with nonlinearities, the local observer is designed by incorporating the available nonlinear model of the leader evaluated at $\varepsilon_i$, yielding
\begin{align}
    &\Dot{\hat{x}}_i = \mathbf{A} \hat{x}_i + B u_i + \frac{1}{2} \mathbf{Q}_i(\varrho_i) C^{\top}( y_i - C\hat{x}_i) + f_0(\varepsilon_i,t), \label{eq:local observer} \\ 
    & \varrho_i(t)=\gamma_{oi}\alpha(t), \label{eq:varrho}
\end{align}
where $\varepsilon_i$ is the local estimate of the leader's state in \eqref{dis_observer}. The admissible range of $\gamma_{oi}$ is
\begin{align}
    \Omega_{oi} = \left\{\gamma_{oi} \in \mathbb{R} | \gamma_{oi} > 2n + 1 -\frac{1}{n}\right\}.
\end{align}
The derivative of local observation error $e_i= x_i - \hat{x}_i$ can be computed from \eqref{eq:system} and \eqref{eq:local observer} as
\begin{align}\label{dot_ei}
    \dot{e}_i = \left(\mathbf{A} - \frac{1}{2}\mathbf{Q}_iC^{\top}C\right) e_i+(f_i-f_{0i}),
\end{align}
where $f_{0i}:=f_0(\varepsilon_i,t)$.
\begin{remark}
The design insight that the nonlinear model $f_0(\cdot,t)$ of leader is available to the followers is consistent with model-based nonlinear observer designs, such as local/global observers, and extended Kalman filters, which commonly exploit the nonlinear dynamics or their nominal models (\cite{khalil2015nonlinear}). In the proposed local observer, $f_0(\varepsilon_i,t)$ only uses the known leader dynamics without requiring global access to the exact leader state. This setting is also practical, since leader or reference dynamics are often obtained from physical modeling or reference-generator design, for example, target motion models in phased-array radar tracking are widely used for filtering, prediction, and beam scheduling (\cite{daum1983decoupled,hao2024nonmyopic}).
\end{remark}

To facilitate a unified stability analysis for the local and distributed observers, define the Lyapunov candidate as $V_{oi} = \varrho_i^{2n+1}  \left(e_i^{\top} \mathbf{Q}_i^{-1} e_i +  \zeta_i^{\top} \mathbf{Q}_i^{-1} \zeta_i\right)$.
The derivative of $V_{oi}$ along the trajectories of \eqref{eq:dev zeta} and \eqref{dot_ei} can be calculated as
\begin{align}
    \begin{aligned}
    \dot{V}_{oi} &= \frac{d(\varrho_i^{2n+1})}{dt}\big(e_i^\top \mathbf{Q}_i^{-1}e_i + \zeta_i^\top \mathbf{Q}_i^{-1}\zeta_i\big) \\
    &\quad + \varrho_i^{2n+1}\left( 2e_i^\top \mathbf{Q}_i^{-1}\dot{e}_i + 2\zeta_i^\top \mathbf{Q}_i^{-1}\dot{\zeta}_i \right) \\
    &\quad + \varrho_i^{2n+1}\left( e_i^\top \dot{\mathbf{Q}_i}^{-1}e_i + \zeta_i^\top \dot{\mathbf{Q}_i}^{-1}\zeta_i \right).
\end{aligned}
\end{align}
Using Lemma \ref{Lemma dQ}, the derivative of the inverse matrix satisfies
$\dot{\mathbf{Q}_i}^{-1} = -\mathbf{Q}_i^{-1}\frac{d\mathbf{Q}_i}{d\varrho_i}\mathbf{Q}_i^{-1}\dot{\varrho_i}
\le -\frac{\varrho_i}{n\gamma_{oi}}\mathbf{Q}_i^{-1}$. Combining with $\dot{\varrho_i} = {\varrho_i^2}/{\gamma_{oi}}$, \eqref{eq:PLE2}, \eqref{eq:dev zeta} and \eqref{dot_ei}, it yields
\begin{align}\label{dot V_oi}
    \begin{aligned}
    \dot{V}_{oi}
    \le&~ \frac{(2n+1)\varrho_i}{\gamma_{oi}} V_{oi} - \varrho_i V_{oi} - \frac{\varrho_i}{n\gamma_{oi}} V_{oi} \\
    &  + 2\varrho_i^{2n+1}e_i^\top \mathbf{Q}_i^{-1}(f_i-f_0)  \\
    =& \left( \frac{2n+1 - 1/n}{\gamma_{oi}} - 1 \right)\varrho_i V_{oi}\\
    &+ 2\varrho_i^{2n+1}e_i^\top \mathbf{Q}_i^{-1}(f_i-f_{0i}).
\end{aligned}
\end{align}

\subsection{Event-triggered output-feedback controller}

To preserve the independent controller--observer parameter selection while reducing computational burden, we design the  event-triggered output-feedback controller as:
\begin{align}
    &u_i = -\frac{1}{2} B^{\top}\bar{\mathbf{P}}_i(r_i)(\bar{\hat{x}}_i - \bar{\varepsilon}_i), \quad t \in [t_k^i, t_{k+1}^i), \label{ui_ETC} \\
     &r_i(t) = \gamma_{ci} \alpha(t), \label{eq:r}
\end{align}
where $\bar{\mathbf{P}}_i(r_i) = \mathbf{P}_i(r_i(t_k^i))$, $\bar{\hat{x}}_i = \hat{x}_i(t_k^i)$, and $\bar{\varepsilon}_i = \varepsilon_i(t_k^i)$ are the triggered variables of corresponding signals. The time-varying parameter is\(r_i(t)\), where \(\gamma_{ci}\in\Omega_{ci}\). The admissible set \(\Omega_{ci}\) is selected according to \eqref{Omega_ic}.
Let $\Delta_{x_i} = \bar{\hat{x}}_i - \hat{x}_i$,  $\Delta_{\varepsilon_i} = \bar{\varepsilon}_i - \varepsilon_i$ and $\Delta_{\mathbf{P}_i} = \mathbf{P}_i - \bar{\mathbf{P}}_i$ be the difference between the triggered and the original signals. The tracking error of follower $i$ is denoted as $\tilde{x}_i = x_i - x_0$. With the fact $\bar{\hat{x}}_i - \bar{\varepsilon}_i = (\hat{x}_i + \Delta_{x_i}) - (\varepsilon_i + \Delta_{\varepsilon_i}) = (\hat{x}_i - \varepsilon_i) + (\Delta_{x_i} - \Delta_{\varepsilon_i})$, the derivative of $\tilde{x}_i$ can be computed from \eqref{eq:system} and \eqref{ui_ETC} as
\begin{align}\label{eq:tracking_error_dynamics}
\dot{\tilde{x}}_i
=&
\left( \mathbf{A} - \frac{1}{2} BB^{\top}\mathbf{P}_i \right)\tilde{x}_i
+ \frac{1}{2} BB^{\top}\mathbf{P}_i(e_i + \zeta_i) \notag\\
& + \frac{1}{2} BB^{\top}\bar{\mathbf{P}}_i(\Delta_{\varepsilon_i} - \Delta_{x_i}) + \frac{1}{2}BB^{\top}\Delta_{\mathbf{P}_i}(\hat{x}_i - \varepsilon_i) \notag\\
&+ (f_i - f_0).
\end{align}

Choose a candidate Lyapunov function as $ V_{ci} = r_i\tilde{x}_i^{\top} \mathbf{P}_i \tilde{x}_i$.
With  \eqref{eq:PLE1},  \eqref{eq:tracking_error_dynamics}, and the fact $\dot{r}_i =  {r_i^2}/{\gamma_{ci}}$,  the derivative of $V_{ci}$ is given by
\begin{align}
\dot V_{ci}
&= \dot r_i\,\tilde x_i^\top \mathbf{P}_i\tilde x_i
+ r_i \dot r_i \,\tilde x_i^\top \frac{d\mathbf{P}_i}{dr_i}\,\tilde x_i - r_i^2\,\tilde x_i^\top \mathbf{P}_i\tilde x_i \notag\\
& \quad
+ r_i\,\tilde x_i^\top \mathbf{P}_iBB^\top \mathbf{P}_i\,(e_i+\zeta_i)
+ 2r_i\,\tilde x_i^\top \mathbf{P}_i\,(f_i-f_0) \notag\\
& \quad+  r_i  \tilde x_i^{\top}\mathbf{P}_iBB^{\top}\bar{\mathbf{P}}_i(\Delta_{\varepsilon_i} - \Delta_{x_i}) \notag\\
&\quad
+  r_i \tilde x_i^{\top}\mathbf{P}_iBB^{\top}\Delta_{\mathbf{P}_i}
(\hat{x}_i - \varepsilon_i).
\label{eq:Vci_dot_step1}
\end{align}
%Using \eqref{eq:4}, it holds that 
%\begin{align}
%    \mathbf{P}_iBB^\top \mathbf{P}_i
%= ( r_i \mathbf{L}_{ci} \mathbf{P}_{ni} \mathbf{L}_{ci} )\, B B^\top \, ( r_i \mathbf{L}_{ci} \mathbf{P}_{ni} \mathbf{L}_{ci} ).
%\end{align}
%From the matrix pencil result in Lemma \ref{lemma_matrix pencil}, it follows that $- \delta_i \mathbf{P}_{ni} + \mathbf{P}_{ni}\mathbf{L}_{ci} B B^{\top} \mathbf{L}_{ci} \mathbf{P}_{ni} \leq 0$, 
%where the constant parameter $\delta_i = \max\{\sigma(\mathbf{P}_{ni} \mathbf{L}_{ci} B B^{\top} \mathbf{L}_{ci} \mathbf{P}_{ni}, \mathbf{P}_{ni})\}$. Then, it follows that 
%\begin{align}\label{delta_i develp}
%    \mathbf{P}_iBB^\top \mathbf{P}_i \leq \delta_i r_i \mathbf{P}_i.
%\end{align}
Together with \eqref{eq:5} in Lemma \ref{lemma dP_i}, \eqref{eq:Vci_dot_step1} can be rewritten as follows.
\begin{align}\label{dot V_ci step 2}
\dot V_{ci}
\leq &\left(\frac{2+\xi_i}{\gamma_{ci}}-1\right) r_i\,V_{ci} + \delta_i r_i^2\,\tilde x_i^{\top}\mathbf{P}_i\,(e_i+\zeta_i)
\notag\\
&
+ 2r_i\,\tilde x_i^{\top}\mathbf{P}_i\,(f_i-f_0) +  r_i  \tilde x_i^{\top}\mathbf{P}_iBB^{\top} \notag \\
& ~~~\times \Big[\underbrace{\bar{\mathbf{P}}_i(\Delta_{\varepsilon_i} - \Delta_{x_i}) +  \Delta_{\mathbf{P}_i}
(\hat{x}_i - \varepsilon_i) }_{\mathbf{M}_i}\Big],
\end{align}
where $\mathbf{M}_i$ is introduced by the event-triggered mechanism.
Applying Young's Inequality, we obtain  
\begin{align*}
     \tilde x_i^{\top}\mathbf{P}_iBB^{\top}\mathbf{M}_i   \leq \frac{1}{2} \left( \tilde x_i^{\top}\mathbf{P}_iBB^{\top}\mathbf{P}_i\tilde x_i + \mathbf{M}_i^{\top}BB^{\top}\mathbf{M}_i\right),
\end{align*}
with the following inequality for $\mathbf{M}_i^{\top}BB^{\top}\mathbf{M}_i$:
\begin{align}\label{eq:MBBM}
    \frac{\mathbf{M}_i^{\top}BB^{\top}\mathbf{M}_i}{2} 
    &\le  (\Delta_{\varepsilon_i} - \Delta_{x_i})^{\top}\bar{\mathbf{P}}_iBB^{\top}\bar{\mathbf{P}}_i(\Delta_{\varepsilon_i} - \Delta_{x_i}) \notag \\
    & \quad +   (\hat{x}_i - \varepsilon_i)^{\top}\Delta_{\mathbf{P}_i}BB^{\top}\Delta_{\mathbf{P}_i}(\hat{x}_i - \varepsilon_i).
\end{align}
Let  
\begin{align}\label{eq:trigger func}
\Psi_i
&=
(\Delta_{\varepsilon_i} - \Delta_{x_i})^{\top}\bar{\mathbf{P}}_iBB^{\top}\bar{\mathbf{P}}_i(\Delta_{\varepsilon_i} - \Delta_{x_i}) \notag\\
&\quad
+
(\hat{x}_i - \varepsilon_i)^{\top}\Delta_{\mathbf{P}_i}BB^{\top}\Delta_{\mathbf{P}_i}(\hat{x}_i - \varepsilon_i).
\end{align}
By applying (\ref{eq:MBBM}) and (\ref{eq:trigger func}),  $\dot V_{ci}$ is simplified to
\begin{align}\label{dot V_ci}
\hspace{-0.2em}\dot V_{ci}
 \leq &\left(\frac{2+\xi_i}{\gamma_{ci}}-1\right) r_i\,V_{ci} + \delta_i r_i^2\,\tilde x_i^{\top}\mathbf{P}_i\,(e_i+\zeta_i)+r_i \Psi_i
\notag\\
&
+ 2r_i\,\tilde x_i^{\top}\mathbf{P}_i\,(f_i-f_0) + \frac{r_i}{2}  \tilde x_i^{\top}\mathbf{P}_iBB^{\top}\mathbf{P}_i\tilde x_i .
\end{align}
Then, the triggering condition is designed as 
\begin{align}\label{triggering condition}
    t_{k+1}^i = \inf \{ t > t_k^i : V_{ci} - 2\Psi_i \le 0 \}.
\end{align}
Once the triggering condition is satisfied,  the controller updates the triggered values to the current signals, i.e., $\bar{\mathbf{P}}_i \leftarrow {\mathbf{P}}_i$, $\bar{\hat{x}}_i\leftarrow {\hat{x}}_i$, and $\bar{\varepsilon}_i \leftarrow {\varepsilon}_i $. Consequently, the triggering errors $\Delta_{\varepsilon_i}$, $\Delta_{x_i}$, and $\Delta_{\mathbf{P}_i}$ are reset to zero. It should be noted that the triggering condition in (\ref{triggering condition}) is associated solely with $V_{ci}$. This specific design choice satisfies the requirement of the separation principle.

\subsection{Consensus tracking without DoS attacks}
Depending on the control scheme developed thus far, the leader-follower consensus tracking in the absence of DoS attacks is established by the following theorem.
% \begin{theorem}\label{Theorem 1}
	%     \textit{(Consensus Tracking in the Absence of DoS Attacks)}  
	%     Consider the MAS modeled as \eqref{eq:system}. Under Assumptions \ref{ass:nonlinearity} and \ref{Ass:Graph}, by implementing the local observer \eqref{eq:local observer}--\eqref{eq:varrho}, the distributed observer \eqref{dis_observer}, and the event-triggered controller \eqref{ui_ETC}--\eqref{eq:r} subject to the triggering condition \eqref{triggering condition}, the following results are guaranteed:
	%     \begin{enumerate}
		%         \item Prescribed-time consensus tracking is achieved in the sense of Definition 1 with the prescribed time $t_f$.
		%         \item There exists some positive constants $\Check{m}_i$, $\Check{M}_i$, $\Check{m}_{u_i}$ and $\Check{M}_{u_i}$ for $\chi_i := [\tilde{x}_i^\top, (e_i + \zeta_i)^\top]^\top$ and $u_i$ such that
		%         \begin{align}
			%             |\chi_i(t)| & \le (t_f-t)^{\Check{m}_i} \Check{M}_i |\chi_i(0)|, \nonumber \\
			%             |u_i(t)| & \le (t_f - t)^{\Check{m}_{u_i}} \Check{M}_{u_i} |\chi_i(0)|. \nonumber
			%         \end{align}
		%         \item Zeno behavior is excluded.
		%     \end{enumerate}
	% \end{theorem}
\begin{theorem}\label{Theorem 1}
	\textit{(A separation principle for consensus tracking in the absence of DoS attacks)}
	Consider the MAS consisting of a leader \eqref{eq:leader} and \(N\) followers \eqref{eq:system} under Assumptions~\ref{ass:nonlinearity} and \ref{Ass:Graph}.   Suppose that the hybrid observer, composed of the distributed observer~\eqref{dis_observer} and the local observer~\eqref{eq:local observer}, together with the event-triggered controller~\eqref{ui_ETC} under the triggering condition~\eqref{triggering condition}, is applied. Then, the following results hold: $i)$ the observer error convergences to zero within the prescribed time $t_f$, and the prescribed-time consensus tracking  is achieved in the sense of Definition~1;  $ii)$ Zeno behavior is excluded; and $iii)$ there exist positive constants $\Check{m}_i$, $\Check{M}_i$,
	$\Check{m}_{u_i}$, and $\Check{M}_{u_i}$ such that, for
	$\chi_i := [\tilde{x}_i^\top, (e_i+\zeta_i)^\top]^\top$,
	\begin{align}
		|\chi_i(t)|
		&\le
		(t_f-t)^{\Check{m}_i}\Check{M}_i|\chi_i(0)|,
		\nonumber\\
		|u_i(t)|
		&\le
		(t_f-t)^{\Check{m}_{u_i}}\Check{M}_{u_i}|\chi_i(0)|.
		\nonumber
	\end{align}
     In addition, the designed output-feedback controller and the hybrid observer satisfy the separation principle in the sense that the coefficients \(\gamma_{ci}\) and \(\gamma_{oi}\) associated with the controller and observer gains can be independently selected from the following two sets:
    \begin{align}\label{Omega_ci oi}
        &\Omega_{ci}
		=
		\left\{
		\gamma_{ci} \in \mathbb{R}
		\,\middle|\,
		\gamma_{ci} > \max\{4+2\xi_i,\ t_f\}
		\right\}, \nonumber \\
      &\Omega_{oi} = \left\{\gamma_{oi} \in \mathbb{R} | \gamma_{oi} > 2n + 1 -\frac{1}{n}\right\}.
    \end{align}
\end{theorem}

\begin{proof}
    The proof is divided into two parts.
    
    \textbf{Part 1}: The overall Lyapunov candidate function for the closed-loop system modeled in \eqref{eq:system} is defined as
    \begin{align}
        V_i&=V_{ci} + c_i V_{oi}, \nonumber\\
        &=r_i\tilde{x}_i^{\top} \mathbf{P}_i \tilde{x}_i + c_i\varrho_i^{2n+1}  \left(e_i^{\top} \mathbf{Q}_i^{-1} e_i +  \zeta_i^{\top} \mathbf{Q}_i^{-1} \zeta_i\right),
    \end{align}
where $c_i$ is a positive constant parameter. 
Combining \eqref{dot V_ci} with \eqref{dot V_oi}, the derivative of $V_i$ yields 
\begin{align}\label{dot_V_con_without reduce}
    \dot{V}_i \leq & \left(\frac{2+\xi_i}{\gamma_{ci}}-1\right) r_i\,V_{ci} + c_i \left( \frac{2n+1 - 1/n}{\gamma_{oi}} - 1 \right)\varrho_i V_{oi} \nonumber  \\
    & + \frac{\delta_i}{2} r_i^2 \tilde x_i^{\top}\mathbf{P}_i \tilde x_i  + \delta_i r_i^2\,\tilde x_i^{\top}\mathbf{P}_i\,(e_i+\zeta_i)  \nonumber\\
&
+ 2r_i\,\tilde x_i^{\top}\mathbf{P}_i\,(f_i-f_0) +r_i \Psi_i\notag\\
&+ 2c_i\varrho_i^{2n+1}e_i^\top \mathbf{Q}_i^{-1}(f_i-f_{0i}).  
\end{align}

\textbf{Step 1.} We first show that $\frac{\delta_i}{2} r_i^2 \tilde x_i^{\top}\mathbf{P}_i \tilde x_i  + \delta_i r_i^2\,\tilde x_i^{\top}\mathbf{P}_i\,(e_i+\zeta_i)$ can be upper bounded by a linear combination of $  V_{ci}$ and $  V_{oi}$.
Recall that $r_i=\gamma_{ci}\alpha$ and $\varrho_i=\gamma_{oi}\alpha$, we can equivalently write the diagonal matrices $\mathbf{L}_{ci}$ and $\mathbf{L}_{oi}$ as 
\begin{align} \label{Lic Lio}
\mathbf{L}_{ci} 
=& \Gamma_{ci} \mathbf{L}_{\alpha}, \qquad   \notag \\
\mathbf{L}_{oi} 
= &\Gamma_{oi} \mathbf{L}_{\alpha}, 
\end{align}
where 
\begin{align}
\Gamma_{ci} &= \mathrm{diag}\{\gamma_{ci}^{n-1},\, \gamma_{ci}^{n-2},\, \ldots,\, 1\}, \notag\\
\Gamma_{oi} &= \mathrm{diag}\{\gamma_{oi}^{n-1},\, \gamma_{oi}^{n-2},\, \ldots,\, 1\}, \notag\\
\mathbf{L}_{\alpha} &= \mathrm{diag}\{\alpha^{n-1},\, \alpha^{n-2},\, \ldots,\, 1\}.
\end{align}
With the fact $\mathbf{P}_i  = r_i \mathbf{L}_{ci} \mathbf{P}_{ni} \mathbf{L}_{ci}$ in (\ref{eq:4}), the coupling term $ r_i^2\,\tilde x_i^{\top}\mathbf{P}_i\,(e_i+\zeta_i)$ satisfies
\begin{align*}
  r_i^2\,\tilde x_i^{\top}\mathbf{P}_i\,(e_i+\zeta_i)
&=   r_i^{3}\tilde{x}_i^{\top} \mathbf{L}_{ci} \mathbf{P}_{ni} \mathbf{L}_{ci} (e_i+\zeta_i) \notag\\
&=   \alpha^{3}\tilde{x}_i^{\top} \mathbf{L}_{\alpha}
\big(\gamma_{ci}^{3}\Gamma_{ci} \mathbf{P}_{ni} \Gamma_{ci} \big)
\mathbf{L}_{\alpha} (e_i+\zeta_i).
\end{align*}
For $V_{ci}$ and $V_{oi}$, we have the following transformations
\begin{align}
r_i V_{ci}
&= \alpha^{3}\tilde{x}_i^{\top}\mathbf{L}_{\alpha}
\big(\gamma_{ci}^{3}\Gamma_{ci}\mathbf{P}_{ni} \Gamma_{ci} \big)
\mathbf{L}_{\alpha}\tilde{x}_i, \label{rVc_i}
\\[2mm]
\varrho_i V_{oi}
&= \alpha^{3}(e_i+\zeta_i)^{\top}\mathbf{L}_{\alpha}
\big(c_i\gamma_{oi}^{3}\Gamma_{oi}\mathbf{Q}_{ni}^{-1}\Gamma_{oi}\big)
\mathbf{L}_{\alpha}(e_i+\zeta_i).\label{rho Voi}
\end{align}
Align with \eqref{rVc_i}, it yields
\[\frac{\delta_i}{2} r_i^2 \tilde x_i^{\top}\mathbf{P}_i \tilde x_i = \frac{\delta_i}{2} \alpha^{3}\tilde{x}_i^{\top}\mathbf{L}_{\alpha}
\big(\gamma_{ci}^{3}\Gamma_{ci} \mathbf{P}_{ni} \Gamma_{ci} \big)
\mathbf{L}_{\alpha}\tilde{x}_i.\]
From Lemma \ref{lemma_matrix pencil}, we can choose $\kappa_i = \max\{\sigma(\mathbf{Q}_{i,\kappa 2}, -\mathbf{Q}_{i,\kappa 1})\}$ with $h_{i} > 0$ such that the matrix pencil holds $\kappa_i \mathbf{Q}_{i,\kappa 1} + \mathbf{Q}_{i,\kappa 2} < 0$,
since $\mathbf{Q}_{i,\kappa 1}$ is SND and $\mathbf{Q}_{i,\kappa 2}$ is symmetric, which are provided in \eqref{eq:Qi_kappa1} and \eqref{eq:Qi_kappa2}.
\begin{align}
\mathbf{Q}_{i,\kappa 1}
&=
\begin{bmatrix}
-\gamma_{ci}^{3}\Gamma_{ci}\mathbf{P}_{ni}\Gamma_{ci} & \mathbf{0} \\
\mathbf{0} & -h_i c_i \gamma_{oi}^{3}\Gamma_{oi}\mathbf{Q}_{ni}^{-1}\Gamma_{oi}
\end{bmatrix},
\label{eq:Qi_kappa1}
\\
\mathbf{Q}_{i,\kappa 2}
&=
\begin{bmatrix}
\frac{\delta_i}{2}\gamma_{ci}^{3}\Gamma_{ci}\mathbf{P}_{ni}\Gamma_{ci} &
\frac{\delta_i}{2}\gamma_{ci}^{3}\Gamma_{ci}\mathbf{P}_{ni}\Gamma_{ci} \\
\frac{\delta_i}{2}\gamma_{ci}^{3}\Gamma_{ci}\mathbf{P}_{ni}\Gamma_{ci} &
\mathbf{0}
\end{bmatrix}.
\label{eq:Qi_kappa2}
\end{align}
Then, the inequality $\frac{\delta_i}{4} r_i^2 \tilde{x}_i^{\top}\mathbf{P}_i \tilde{x}_i
+ \delta_i r_i^2 \tilde{x}_i^{\top}\mathbf{P}_i(e_i+\zeta_i)
\le
\kappa_i \big(r_i V_{ci}\big)
+ h_i \kappa_i\big(\varrho_i V_{oi}\big)$ can be regarded as a direct consequence of the quadratic form of
the matrix pencil. In particular, with respect to
$\alpha[\tilde{x}_i^{\top}\mathbf{L}_{\alpha},
(e_i+\zeta_i)^{\top}\mathbf{L}_{\alpha}]^{\top}$, the corresponding
quadratic form is given by
\begin{align}
&-\kappa_i \alpha^3 \tilde{x}_i^{\top}\mathbf{L}_{\alpha}
\bar{\mathbf P}_{ci}
\mathbf{L}_{\alpha}\tilde{x}_i
-\kappa_i \alpha^3(e_i+\zeta_i)^{\top}\mathbf{L}_{\alpha}
\bar{\mathbf Q}_{oi}
\mathbf{L}_{\alpha}(e_i+\zeta_i)
\nonumber\\
&\quad
+\frac{\delta_i}{2}\alpha^3\tilde{x}_i^{\top}\mathbf{L}_{\alpha}
\bar{\mathbf P}_{ci}
\mathbf{L}_{\alpha}\tilde{x}_i
+\delta_i\alpha^3\tilde{x}_i^{\top}\mathbf{L}_{\alpha}
\bar{\mathbf P}_{ci}
\mathbf{L}_{\alpha}(e_i+\zeta_i)
\le 0.\nonumber
\end{align}

\textbf{Step 2.} The coupling terms involving nonlinear functions $+ 2r_i\,\tilde x_i^{\top}\mathbf{P}_i\,(f_i-f_0) + 2c_i\varrho_i^{2n+1}e_i^\top \mathbf{Q}_i^{-1}(f_i-f_{0i})$ can be upper bounded by a linear combination of $V_{ci}$ and $V_{oi}$. 

Under Assumption \ref{ass:nonlinearity}, 
using \eqref{eq:4} and (\ref{eq:lipschitz_phi}), we have 
\begin{align}
2r_i\tilde{x}_i^{\top}&\mathbf{P}_i(f_i-f_0)
\le 2r_i \lvert \tilde{x}_i \rvert_e^{\top}
      \lvert \mathbf{P}_i \rvert_e
      \lvert f_i - f_0 \rvert_e \notag\\
&\le 2r_i \theta_i \lvert \tilde{x}_i \rvert_e^{\top}
      \big( r_i \mathbf{L}_{ci} \lvert \mathbf{P}_{ni} \rvert_e \mathbf{L}_{ci} \big)
      \mathbf{A}_f
      \lvert \tilde{x}_i \rvert_e \notag\\
&= 2\theta_i r_i^2 \lvert \tilde{x}_i \rvert_e^{\top}
     \mathbf{L}_{ci} \lvert \mathbf{P}_{ni} \rvert_e
     \big( \mathbf{L}_{ci} \mathbf{A}_f \mathbf{L}_{ci}^{-1} \big)
     \mathbf{L}_{ci}
     \lvert \tilde{x}_i \rvert_e .
\end{align}
From \eqref{Omega_ic}, we have that $r_i(t) \ge 1$, which implies $\mathbf{L}_{ci} \mathbf{A}_f \le_e \mathbf{A}_f \mathbf{L}_{ci}$. With \eqref{Lic Lio}, it yields  $$2r_i\tilde{x}_i^{\top}\mathbf{P}_i(f_i-f_0) \le 2\theta_i \gamma_{ci}^2 \alpha^2 \lvert \tilde{x}_i \rvert_e^{\top} \mathbf{L}_\alpha \Gamma_{ci} \lvert \mathbf{P}_{ni} \rvert_e \mathbf{A}_f \Gamma_{ci} \mathbf{L}_\alpha \lvert \tilde{x}_i \rvert_e.$$

Since $x_i-\varepsilon_i = x_i - x_0 - (\varepsilon_i - x_0) = \tilde{x}_i - \zeta_i$, based on the property of $\lvert \cdot \rvert_e$, we have that $\lvert x_i-\varepsilon_i \rvert_e = \lvert \tilde{x}_i - \zeta_i \rvert_e \le \lvert \tilde{x}_i \rvert_e + \lvert \zeta_i \rvert_e$. Similarly, with $\mathbf{Q}_i^{-1}(\varrho_i) = \varrho_i^{-(2n-1)} \mathbf{L}_{oi} \mathbf{Q}_{ni}^{-1} \mathbf{L}_{oi}$ and Assumption \ref{ass:nonlinearity}, it holds that 
\begin{align}
2&c_i\varrho_i^{2n+1} e_i^{\top}\mathbf{Q}_i^{-1}(f_i-f_{0i})
= 2c_i\varrho_i^{2} e_i^{\top}\mathbf{L}_{oi} \mathbf{Q}_{ni}^{-1}\mathbf{L}_{oi} (f_i-f_{0i}) \notag\\
&\le 2c_i\varrho_i^{2}\lvert e_i\rvert_e^{\top}
\mathbf{L}_{oi} \lvert \mathbf{Q}_{ni}^{-1}\rvert_e \mathbf{L}_{oi}
\big(\theta_i \mathbf{A}_f \lvert x_i-\varepsilon_i \rvert_e  \big) \\
&\le 2c_i\varrho_i^{2}\lvert e_i\rvert_e^{\top}
\mathbf{L}_{oi} \lvert \mathbf{Q}_{ni}^{-1}\rvert_e \mathbf{L}_{oi}
\big(\theta_i \mathbf{A}_f(\lvert \tilde{x}_i \rvert_e + \lvert \zeta_i \rvert_e) \big) \notag\\
&= 2c_i\theta_i \gamma_{oi}^{2}\alpha^{2}
\lvert e_i\rvert_e^{\top}
\mathbf{L}_{\alpha} \Gamma_{oi}
\lvert \mathbf{Q}_{ni}^{-1}\rvert_e
\mathbf{A}_f
\Gamma_{oi} \mathbf{L}_{\alpha}
(\lvert \tilde{x}_i \rvert_e + \lvert \zeta_i \rvert_e).\notag
\end{align}
Define two positive definite matrices
$\bar{\mathbf{P}}_{ni} = \mathrm{diag} \{\mathbf{P}_{ni}\}$ and $\bar{\mathbf{Q}}_{ni}^{-1}= \mathrm{diag} \{\mathbf{Q}_{ni}^{-1}\}$. Since $\mathbf{P}_{ni}$ and $\mathbf{Q}_{ni}^{-1}$ are SPD matrices, the following
generalized eigenvalue relationships hold
\begin{equation}
\begin{aligned}
-\mathbf{P}_{ni} + \delta_{i}^P \bar{\mathbf{P}}_{ni} &< 0, \\
-\mathbf{Q}_{ni}^{-1} + \delta_{i}^Q\bar{\mathbf{Q}}_{ni}^{-1} &< 0 ,
\end{aligned}
\end{equation}
where
\begin{equation}
\begin{aligned}
\delta_{i}^P&=\min\{\sigma(\mathbf{P}_{ni},\bar{\mathbf{P}}_{ni})\}, \\
\delta_{i}^Q&=\min\{\sigma(\mathbf{Q}_{ni}^{-1},\bar{\mathbf{Q}}_{ni}^{-1}\}.
\end{aligned}
\end{equation}
With Lemma~\ref{lemma_matrix pencil}, the matrix pencil satisfies $s_i \mathbf{Q}_{i,s1} + \mathbf{Q}_{i,s2}<0$, 
where $s_i=\max\{\sigma(\mathbf{Q}_{i,s2},-\mathbf{Q}_{i,s1})\}$. Since $\mathbf{Q}_{i,s1}$ is SND and $\mathbf{Q}_{i,s2}$ is symmetric, where two matrices are given as \eqref{eq:si_matrices}.

\begin{figure*}[!t]
\footnotesize
\setlength{\arraycolsep}{3pt}
\hrule
\vspace{0.3em}

\begin{equation}
\begin{aligned}
&s_i
\underbrace{
\left[
\begin{array}{cc}
-\delta_i^P\gamma_{ci}^{2}
\Gamma_{ci}|\bar{\mathbf P}_{ni}|_e\Gamma_{ci}
&
\mathbf 0
\\[1mm]
\mathbf 0
&
-c_i\delta_i^Q\gamma_{oi}^{2}
\Gamma_{oi}|\bar{\mathbf Q}_{ni}^{-1}|_e\Gamma_{oi}
\end{array}
\right]
}_{\mathbf Q_{i,s1}}
+
\underbrace{
\left[
\begin{array}{cc}
\theta_i\gamma_{ci}^{2}
\left(
\Gamma_{ci}^{\top}|\bar{\mathbf P}_{ni}|_e\mathbf A_f\Gamma_{ci}
+
\Gamma_{ci}\mathbf A_f^{\top}|\bar{\mathbf P}_{ni}|_e\Gamma_{ci}
\right)
&
c_i\theta_i\gamma_{oi}^{2}
\Gamma_{oi}\mathbf A_f^{\top}
|\bar{\mathbf Q}_{ni}^{-1}|_e\Gamma_{oi}
\\[1mm]
c_i\theta_i\gamma_{oi}^{2}
\Gamma_{oi}|\bar{\mathbf Q}_{ni}^{-1}|_e
\mathbf A_f\Gamma_{oi}
&
\mathbf 0
\end{array}
\right]
}_{\mathbf Q_{i,s2}}
<0 .
\end{aligned}
\label{eq:si_matrices}
\end{equation}

\begin{equation}
\begin{aligned}
&-s_i \Bigg[
\delta_i^P\alpha^2
|\tilde{x}_i|_e^{\top}
\mathbf L_{\alpha}
\big(
\gamma_{ci}^{2}\Gamma_{ci}
|\bar{\mathbf P}_{ni}|_e\Gamma_{ci}
\big)
\mathbf L_{\alpha}
|\tilde{x}_i|_e
+
\delta_i^Q\alpha^2
|e_i+\zeta_i|_e^{\top}
\mathbf L_{\alpha}
\big(
c_i\gamma_{oi}^{2}\Gamma_{oi}
|\bar{\mathbf Q}_{ni}^{-1}|_e\Gamma_{oi}
\big)
\mathbf L_{\alpha}
|e_i+\zeta_i|_e
\Bigg]
\\
&\quad
+2\theta_i r_i^2 |\tilde{x}_i|_e^{\top}
\mathbf L_{ci}|\mathbf P_{ni}|_e
\big(
\mathbf L_{ci}\mathbf A_f\mathbf L_{ci}^{-1}
\big)
\mathbf L_{ci}|\tilde{x}_i|_e
+2c_i\theta_i \gamma_{oi}^{2}\alpha^{2}
|e_i|_e^{\top}
\mathbf L_{\alpha}\Gamma_{oi}
|\mathbf Q_{ni}^{-1}|_e
\mathbf A_f
\Gamma_{oi}\mathbf L_{\alpha}
\big(
|\tilde{x}_i|_e+|\zeta_i|_e
\big)
\le 0 .
\end{aligned}
\label{eq:si_cross_condition}
\end{equation}

\vspace{0.2em}
\hrule
\vspace{0.5em}
\end{figure*}
With diagonal matrices $\bar{\mathbf{P}}_{ni}$ and $\bar{\mathbf{Q}}_{ni}^{-1}$, it follows that
\begin{align}
y^\top\bar{\mathbf{P}}_{ni} y
&=
|y|_e^\top|\bar{\mathbf{P}}_{ni}|_e |y|_e, \nonumber\\
y^\top \bar{\mathbf{Q}}_{ni}^{-1} y
&=
|y|_e^\top |\bar{\mathbf{Q}}_{ni}^{-1}|_e |y|_e.
\end{align}
Since
\begin{align}
V_{ci}
&\ge
\delta_{i}^P \alpha^{2}
\lvert \tilde{x}_i \rvert_e^{\top}
\mathbf{L}_{\alpha}
\big(
\gamma_{ci}^{2}\Gamma_{ci}
\lvert \bar{\mathbf{P}}_{ni} \rvert_e
\Gamma_{ci}
\big)
\mathbf{L}_{\alpha}
\lvert \tilde{x}_i \rvert_e,
\label{eq:Vci_lower_bound}
\end{align}
and
\begin{align}
V_{oi}
&\ge
\delta_{i}^Q \alpha^{2}
\lvert e_i+\zeta_i \rvert_e^{\top}
\mathbf{L}_{\alpha}
\big(
c_i\gamma_{oi}^{2}\Gamma_{oi}
\lvert \bar{\mathbf{Q}}_{ni}^{-1} \rvert_e
\Gamma_{oi}
\big)
\mathbf{L}_{\alpha}
\lvert e_i+\zeta_i \rvert_e,
\label{eq:Voi_lower_bound}
\end{align} 
Similarly, by applying the quadratic form associated with the matrix pencil in
\eqref{eq:si_cross_condition}, we obtain
\begin{align}
2r_i\,\tilde{x}_i^{\top}\mathbf{P}_i(f_i-f_0)
+ 2c_i\varrho_i^{2n+1}e_i^\top \mathbf{Q}_i^{-1}(f_i-f_{0i})
\leq s_i V_i .
\end{align}

\textbf{Step 3.} By the triggering condition \eqref{triggering condition}, it follows that $2\Psi_i \leq V_{ci} $. Then, for the term involving triggering function, we have that $r_i \Psi_i \le  \frac{1}{2}r_i V_{ci}$.

Based on the above discussion, \eqref{dot_V_con_without reduce} can be rewritten as
\begin{align}\label{Vi_dot middle}
    \dot{V}_i \leq -\kappa_{ci} r_i\,V_{ci} -\kappa_{oi} \varrho_i V_{oi} + s_i V_i,
\end{align}
where 
\begin{align}
\kappa_{ci}
&=
\frac{1}{2}
-
\frac{2+\xi_i}{\gamma_{ci}}
-
\kappa_i
>0,
\label{eq:kappa_ci}\\
\kappa_{oi}
&=
c_i
\left(
1-\frac{2n+1-1/n}{\gamma_{oi}}
\right)
-
h_i\kappa_i
>0.
\label{eq:kappa_oi}
\end{align}
Rearrange \eqref{Vi_dot middle}, it yields
\begin{align}\label{dot V con reduce}
    \dot{V}_i \leq -\Upsilon_i \alpha(t)V_i + s_i V_i,
\end{align}
where $\Upsilon_i = \min\{\kappa_{ci}\gamma_{ci},  \frac{\kappa_{oi} \gamma_{oi}}{c_i} \}$. Based on Lemma 1 in \cite{ref4}, we can obtain that
\begin{equation} \label{Vi upper bound}
V_i(t)\le \left(\frac{t_f-t}{t_f}\right)^{\Upsilon_i} e^{s_it_f}V_i(0),
\qquad t\in[0,t_f).
\end{equation}
Then, it follows that $V_i(t) \geq \alpha^2(t)\phi_{i,1}(\alpha_0)|\chi_i(t)|^2$, where $\alpha_0=\alpha(0)$ and
\begin{equation*}
\begin{aligned}
\phi_{i,1}(\alpha_0)
&=
\min\{\gamma_{ci}^2,\,c_i\gamma_{oi}^2\}
\underline{\lambda}_i(\alpha_0),\\
\underline{\lambda}_i(\alpha_0)
&=
\min\Big\{
\lambda_{\min}\!\big(\mathbf{L}_{ci}(\alpha_0)\mathbf{P}_{ni}
\mathbf{L}_{ci}(\alpha_0)\big),\\
&\qquad\quad
\lambda_{\min}\!\big(\mathbf{L}_{oi}(\alpha_0)\mathbf{Q}_{ni}^{-1}
\mathbf{L}_{oi}(\alpha_0)\big)
\Big\}.
\end{aligned}
\end{equation*}
Similarly, for $V_i(0)$, one has $V_i(0) \leq \phi_{i,2}(\alpha_0)|\chi_i(0)|^2$, 
where
\begin{equation*}
\begin{aligned}
\phi_{i,2}(\alpha_0)
&=
\max\{\gamma_{ci}^2,\,c_i\gamma_{oi}^2\}
\overline{\lambda}_i(\alpha_0),\\
\overline{\lambda}_i(\alpha_0)
&=
\max\Big\{
\lambda_{\max}\!\big(\mathbf{L}_{ci}(\alpha_0)\mathbf{P}_{ni}
\mathbf{L}_{ci}(\alpha_0)\big),\\
&\qquad\quad
\lambda_{\max}\!\big(\mathbf{L}_{oi}(\alpha_0)\mathbf{Q}_{ni}^{-1}
\mathbf{L}_{oi}(\alpha_0)\big)
\Big\}.
\end{aligned}
\end{equation*}
Let $\Check{m}_i =  \frac{\Upsilon_i}{2} + 1$ and $\Check{M}_i = \sqrt{\frac{\phi_{i,2}(\alpha_0)}{t_f^{\Upsilon_i} \phi_{i,1}(\alpha_0) }}$, the following upper bound holds for the stack vector $\chi_i$
\begin{align}\label{chi bound}
    |\chi_i(t)| &\le \sqrt{\frac{(t_f-t)^{\Upsilon_i + 2}\phi_{i,2}(\alpha_0)}{t_f^{\Upsilon_i}\phi_{i,2}(\alpha_0)}|\chi_i(0)|^2} \nonumber \\
    & \le (t_f-t)^{\Check{m}_i} \Check{M}_i |\chi_i(0)|.
\end{align}
Then, the consensus tracking of the MAS is achieved in the sense of Definition 1.
For the control input, \eqref{ui_ETC} can be rewritten as
\begin{align}
u_i
=-\frac12 B^\top \mathbf P_i\big(\tilde{x}_i-(e_i+\zeta_i)\big)
+\frac12 B^\top \mathbf M_i,
\end{align}
which follows that
\begin{align} \label{upper ui 1}
|u_i|
\le \frac12 \big|B^\top \mathbf P_i(\tilde{x}_i-(e_i+\zeta_i))\big|
+\frac12 |B^\top \mathbf M_i|.
\end{align}
By applying the triangle inequality, the right-hand side of (\eqref{upper ui 1}) can be bounded by using $V_i$.
It follows from \eqref{eq:MBBM} and \eqref{eq:trigger func} that $\mathbf M_i^\top B B^\top \mathbf M_i \le 2\Psi_i$.
Under the event-triggering condition \eqref{triggering condition}, it guarantees that $\Psi_i \le V_{ci} \le V_i$. This directly implies $\frac{1}{2} |B^\top \mathbf M_i| \le \frac{\sqrt{2}}{2}\sqrt{V_i}$.
Using \eqref{eq:5}, it yields $|B^\top \mathbf P_i \tilde{x}_i|^2 \le \delta_i V_{ci} \le \delta_i V_i$. 
Similarly, we have $|B^\top  \mathbf P_i (e_i+\zeta_i)|^2 \le \delta_i r_i (e_i+\zeta_i)^\top \mathbf P_i (e_i+\zeta_i)$.
Recalling that $r_i P_i = \gamma_{ci}^2 \alpha^2 \mathbf L_\alpha \Gamma_{ci} \mathbf P_{ni} \Gamma_{ci} \mathbf L_\alpha$ and $V_{oi} = \alpha^2(e_i + \zeta_i)^\top \mathbf L_\alpha (c_i \gamma_{oi}^2 \Gamma_{oi} \mathbf Q_{ni}^{-1} \Gamma_{oi}) \mathbf L_\alpha (e_i + \zeta_i)$, both expressions are quadratic forms in terms of $\alpha \mathbf L_\alpha (e_i + \zeta_i)$. By virtue of Lemma \ref{lemma_matrix pencil}, we can select a constant
\begin{align}
    \delta_{ui}
    =
    \max\left\{
    \sigma(
    \delta_i \gamma_{ci}^2 \Gamma_{ci} \mathbf P_{ni} \Gamma_{ci},
    c_i \gamma_{oi}^2 \Gamma_{oi} \mathbf Q_{ni}^{-1} \Gamma_{oi}
    )
    \right\}.
\end{align}
such that
$
\delta_i r_i (e_i + \zeta_i)^\top P_i (e_i + \zeta_i) \le \delta_{ui} V_{oi} \le \delta_{ui} V_i
$.
Substituting these upper bounds back into \eqref{upper ui 1}, the control input is bounded by
$
    |u_i(t)| \le \frac{1}{2} \big(\sqrt{\delta_i} + \sqrt{\delta_{ui}} + \sqrt{2}\big) \sqrt{V_i(t)}.
$
According to \eqref{Vi upper bound}, it can be concluded that there exist positive constants $\Check{m}_{u_i}$ and $\Check{M}_{u_i}$ such that
\begin{align} \label{ui bound}
    |u_i(t)| \le (t_f - t)^{\Check{m}_{u_i}} \Check{M}_{u_i} |\chi_i(0)|,
\end{align}
where $\Check{m}_{u_i} = \Upsilon_i / 2$ and 

\noindent$\Check{M}_{u_i} = \frac{1}{2} \big(\sqrt{\delta_i} + \sqrt{\delta_{ui}} + \sqrt{2}\big) \sqrt{e^{s_i t_f} \phi_{i,2}(\alpha_0) / \phi_{i,1}(\alpha_0)}$.

\textbf{Part 2.} The exclusion of Zeno behavior is proved by showing that there exists a strictly positive minimum inter-event time $\tau_i = t_{k+1}^i - t_k^i > 0$ for any agent $i$ during the prescribed time $t \in [0, t_f)$. For any time interval $t \in [t_k^i, t_{k+1}^i) \subset [0, t_f)$, the event-triggering condition guarantees that $\Psi_i(t) \le V_{ci}(t)$.

Note that during $t\in[t_{k}^{i},t_{k+1}^{i})$, the triggered values $\hat{\bar{x}}_{i}$, $\bar{\varepsilon}_{i}$, and $\bar{\mathbf{P}}_{i}$ remain constant, which implies $\dot{\Delta}_{x_{i}}=-\dot{\hat{x}}_{i}$, $\dot{\Delta}_{\varepsilon_{i}}=-\dot{\varepsilon}_{i}$, and $\dot{\Delta}_{\mathbf{P}_{i}}=\dot{\mathbf{P}}_{i}$. For analytical clarity, define the observation mismatch as $\nu_{i}=\hat{x}_{i}-\varepsilon_{i}$. Since $\hat{x}_{i}-\varepsilon_{i}=(x_{i}-e_{i})-(x_{0}+\zeta_{i})=\tilde{x}_{i}-(e_{i}+\zeta_{i})$, we have $\dot{\nu}_{i}=\dot{\hat{x}}_{i}-\dot{\varepsilon}_{i}$. Taking the time derivative of $\Psi_i$ yields
\begin{align}\label{eq:psi_dot}
    \dot{\Psi}_i =&  2(\Delta_{\varepsilon_i} - \Delta_{x_i})^{\top}\bar{\mathbf{P}}_i BB^\top \bar{\mathbf{P}}_i \dot{\nu}_i  + 2\nu_i \Delta_{\mathbf{P}_i} BB^\top \Delta_{\mathbf{P}_i} \dot{\nu}_i \notag \\
    & + 2\nu_i^{\top} \frac{d}{dt} (\Delta_{\mathbf{P}_i} BB^\top \Delta_{\mathbf{P}_i}) \nu_i \notag \\
   \leq &  2(\Delta_{\varepsilon_i} - \Delta_{x_i})^{\top}\bar{\mathbf{P}}_i BB^\top \bar{\mathbf{P}}_i \dot{\nu}_i + 2r_i \nu_i^\top \delta_i'\Delta_{\mathbf{P}_i} \dot{\nu}_i\notag \\
   & + \frac{2\delta_i'}{\gamma_{ci}} \nu_i ^{\top} r_i^2\Delta_{\mathbf{P}_i}\nu_i + c_{pi} \nu_i ^{\top} r_i^2 \mathbf{P}_i\nu_i,
\end{align}
where $\delta_i' = \max \{ \sigma( \Delta_{\mathbf{P}_i} BB^\top \Delta_{\mathbf{P}_i}, r_i \Delta_{\mathbf{P}_i} ) \}$  and $c_{pi} = \frac{2\delta_i'(1+\xi_i)}{\gamma_{ci}}$. 

From the definition of $\chi_i$, it shares the same convergence rate of $\nu_i$. Since $V_i$ is upper bounded, the forth term is strictly bounded. By selecting $\gamma_{oi}$ and $\gamma_{ci}$ according to \eqref{Omega_ci oi} and \eqref{Omega_ic}, $\Upsilon_i > 1$ can be easily guaranteed, which implies that \eqref{Vi upper bound} can be rewritten as $V_i(t) \leq (t_f - t)^{\Upsilon_i} Z_i V_i(0)$, where $Z_i = {e^{s_i t_f}}/{t_f^{\Upsilon_i}}$ is a constant. Since $r_i = \frac{\gamma_{ci}}{t_f - t}$, it can conclude that $r_i V_i \leq (t_f - t)^{\Upsilon_i-1} Z_i' V_i(0)$, where $Z_i' = \gamma_{ci} Z_i$ is a constant. As a result, the third term is upper bounded. 
 From \eqref{eq:Vci_dot_step1}, one has
\begin{align}
\dot V_{ci}
&\le
\dot r_i\,\tilde x_i^\top \mathbf P_i\tilde x_i
+
r_i\,\tilde x_i^\top
\frac{d\mathbf P_i}{dt}
\tilde x_i
+
2r_i\,\tilde x_i^\top \mathbf P_i\dot{\tilde x}_i \notag\\
&=
\frac{2+\xi_i}{\gamma_{ci}}r_iV_{ci}
+
2r_i\,\tilde x_i^\top \mathbf P_i\dot{\tilde x}_i \notag\\
&\le
-\kappa_{ci}r_iV_{ci}
+
s_iV_i .
\label{eq:Vci_dot_bound_final}
\end{align} 
It follows from \eqref{eq:Vci_dot_bound_final} that
\begin{equation}
2r_i\,\tilde x_i^\top \mathbf P_i\dot{\tilde x}_i
\le
-\left(
\kappa_{ci}
-
\frac{2+\xi_i}{\gamma_{ci}}
\right)r_iV_{ci}
+
s_iV_i .
\label{eq:cross_term_bound}
\end{equation}
Since $V_{ci}$ and $V_i$ are bounded, one further obtains $2r_i\,\tilde x_i^\top \mathbf P_i\dot{\tilde x}_i
\in\mathcal L_{\infty}$.
From \eqref{chi bound}, $\tilde{x}_i$ shares the same convergence rate with $(e_i + \zeta_i)$, we can obtanin that the third term is upper bounded. Rewriting the derivative of the tracking and observation errors yields $\dot{\nu}_i = \dot{\tilde{x}}_i - (\dot{e}_i + \dot{\zeta}_i)$. The magnitude of $\dot{\nu}_i$ is governed by a linear combination of the states coupled with the dynamic gains and the bounded nonlinearity $f_i - f_0$. For the first term, the triggering error $(\Delta_{\varepsilon_i} - \Delta_{x_i})$ is inherently bounded by $V_{ci}$ due to the event-triggering condition $\Psi_i(t) \le V_{ci}(t)$. Since $\bar{\mathbf{P}}_i$ is a constant matrix over $[t_k^i, t_{k+1}^i)$, the term $2(\Delta_{\varepsilon_i} - \Delta_{x_i})^{\top}\bar{\mathbf{P}}_i BB^\top \bar{\mathbf{P}}_i \dot{\nu}_i$ is uniformly bounded.

Based on the above discussion, there exists a finite constant $\Theta_i > 0$ such that
\begin{align}
    \dot{\Psi}_i(t) \le \Theta_i, \quad \forall t \in [t_k^i, t_{k+1}^i).
\end{align}

By integrating this inequality over the time elapsed since the last trigger, the accumulation of $\Psi_i(t)$ is constrained by $\Psi_i(t) \le \Theta_i(t - t_k^i)$. For the next event to be triggered at $t_{k+1}^i$, the triggering function must evolve from zero to intersect with the dynamic threshold, satisfying the boundary condition $\Psi_i(t_{k+1}^i) = \frac{1}{2} V_{ci}(t_{k+1}^i)$. Substituting this condition into the integration inequality yields $\Theta_i(t_{k+1}^i - t_k^i) \ge \Psi_i(t_{k+1}^i) =  \frac{1}{2}V_{ci}(t_{k+1}^i)$.
Thus, the minimum time required to cross the triggering threshold is governed by:
\begin{align}
    t_{k+1}^i - t_k^i \ge \frac{V_{ci}(t_{k+1}^i)}{2\Theta_i} := \tau_i > 0.
\end{align}
With $\Theta_i$ being finite, the inter-event time $\tau_i$ is strictly greater than zero. Therefore, Zeno behavior is successfully excluded over the prescribed time interval.
\end{proof}
\begin{table}[!ht]
\centering
\caption{Reasonable ranges of $\gamma_{ci}$ and $\gamma_{oi}$ under different system orders}
\label{tab:parameter_ranges}
\renewcommand{\arraystretch}{1.15}
\setlength{\tabcolsep}{2.2pt}
\scriptsize
\begin{tabular}{cccc}
\toprule
Order & $\xi_i$ & $\gamma_{ci}$ & $\gamma_{oi}$ \\
\midrule
$1$ 
& $0$      
& $(\max\{4,t_f\},\infty)$        
& $(2,\infty)$ \\

$2$ 
& $2.4142$ 
& $(\max\{8.8284,t_f\},\infty)$   
& $(4.5,\infty)$ \\

$3$ 
& $5.2899$ 
& $(\max\{14.5798,t_f\},\infty)$  
& $(6.6667,\infty)$ \\

$\vdots$ 
& $\vdots$ 
& $\vdots$ 
& $\vdots$ \\

$n$ 
& \makecell{$\max\{\sigma(\mathbf{E}_i\mathbf{P}_{ni}$\\
$+\mathbf{P}_{ni}\mathbf{E}_i,$\\
$\mathbf{P}_{ni})\}$}
& \makecell{$(\max\{4+2\xi_i,$\\
$t_f\},\infty)$}
& \makecell{$(2n+1-\frac{1}{n},$\\
$\infty)$} \\
\bottomrule
\end{tabular}
\end{table}

The choices of $\gamma_{ci}$ and $\gamma_{oi}$ satisfy the separation principle as illustrated in Theorem 1, where the explicit admissible ranges of the design coefficients are summarized in Table~\ref{tab:parameter_ranges}. 
% In particular, $\gamma_{ci}$ only affects the positivity of $\kappa_{ci}$, while $\gamma_{oi}$ only affects the positivity of $\kappa_{oi}$. The explicit admissible ranges of the design coefficients are summarized in Table~\ref{tab:parameter_ranges}.
\begin{remark}
Compared with \cite{ref4}, the consensus tracking of nonlinear MASs involves several design and analysis difficulties. Under the directed graph in Assumption~\ref{Ass:Graph}, the leader state is not directly available for all followers, and thus the relative tracking error $x_i-x_0$ cannot be directly reconstructed from a single local estimate. The proposed hybrid observer addresses this issue by generating $\hat{x}_i-\varepsilon_i$ as an implementable estimate of the tracking error. This design introduces the local observation error and the distributed observation error into the closed-loop tracking dynamics, which makes the unified matrix-pencil-based decoupling analysis more involved. Moreover, the local observer incorporates $f_0(\varepsilon_i,t)$ to compensate for the leader nonlinearity, which is necessary for nonlinear leader-following tracking. Finally, the event-triggered errors are handled through a triggering condition depending only on $V_{ci}$, so that no extra constraint is imposed on $\gamma_{oi}$ and the admissible ranges of $\gamma_{ci}$ and $\gamma_{oi}$ remain independent.
\end{remark}

\section{From attack-free to DoS attacks}
\subsection{Denial-of-Service Attacks}
This paper considers the case where all communication links are interrupted during DoS-active intervals, which captures intermittent communication caused by cyber attacks or adverse physical environments. For instance, in autonomous aerial vehicle (AAV) networks, agents may temporarily lose communication with all neighbors when they enter occluded regions, such as areas blocked by buildings or other obstacles. Since such interruptions are usually constrained by the attacker’s energy, bandwidth, or the duration of adverse environmental conditions, the DoS attacks are modeled by imposing standard frequency and duration constraints. Let $\{d_l\}_{l\in\mathbb{N}}$ denote the sequence of DoS off/on transition instants, $\tau_l \in \mathbb{R}_{\geq 0}$ denote the duration of the $l$th DoS attack, and $\mathcal{D}_l = [d_l,d_l+\tau_l)$ denote the corresponding attack interval. Furthermore, for any $t\ge 0$, define the set of DoS-active intervals over $[0,t]$ as
$
\Phi_{\mathcal D}(0,t)=\bigcup_{l\in\mathbb{N}}\left(\mathcal{D}_l\cap[0,t]\right),
$
and define the set of DoS-free intervals as
$
\Phi_{\mathcal H}(0,t)=[0,t]\setminus \Phi_{\mathcal D}(0,t).
$
Moreover, DoS attacks in this paper satisfy the following widely adopted assumption:
\begin{assumption}(\cite{ref-DOS})\label{ass_dos}
For any $t\ge 0$, the DoS frequency and duration satisfy
\begin{align}
    &n_d(0,t)\le \frac{t}{\tau_f}+k_{df},
    \label{eq:dos_frequency} \\
    &T_d(0,t)\le \frac{t}{\tau_d}+k_{dT},
    \label{eq:dos_duration}
\end{align}
where $\tau_f>0$, $\tau_d>1$, $k_{df}\ge 0$, and $k_{dT}\ge 0$ are constants. Here, $n_d(0,t)$ denotes the number of DoS off/on transitions over $[0,t]$, and $T_d(0,t)=|\Phi_{\mathcal D}(0,t)|$ denotes the total duration of DoS intervals over $[0,t]$.
\end{assumption}

\begin{remark}
The constants $k_{df}$ and $k_{dT}$ describe the possible initial attack capability over a finite interval. The parameter $\tau_f$ gives an average lower bound on the interval between two consecutive DoS activations, while $\tau_d$ bounds the average proportion of time during which the network is under DoS attacks. 
\end{remark}

\subsection{Consensus tracking under DoS attacks}
When DoS attacks are activated, all communication links are interrupted as shown in Fig. \ref{fig:digraph}, and the neighbors' information becomes unavailable. To maintain the estimation of the leader's state, the distributed observer switches to an autonomous prediction mode. The distributed observer \eqref{dis_observer} is modified as follows
\begin{equation}
\label{eq:switched_observer}
\dot{\epsilon}_{i} = \begin{cases}
\frac{\sum_{j=0}^{N}a_{ij}\left[\dot{\epsilon}_{j} - \frac{\varrho_i}{2}(\epsilon_{i}-\epsilon_{j})\right]}{\sum_{j=0}^{N}a_{ij}}, & t \in \Phi_{\mathcal{H}} \\
\mathbf{A}\epsilon_{i} +f_0, & t \in \Phi_{\mathcal{D}}
\end{cases}
\end{equation}
For $t \in \Phi_{\mathcal{D}}$, the derivative of $\zeta_i$ is
\begin{equation} \label{new zeta i}
\dot{\zeta}_i = \dot{\epsilon}_i - \dot{x}_0 = \mathbf{A}\zeta_i .
\end{equation}

In \cite{ref-Fenggang-Auto}, the authors consider a scenario where a DoS attack occurs exactly at $t_f$, preventing consensus tracking from being achieved at the originally prescribed instant, as illustrated in Fig. \ref{fig:DoS_time}. 
\begin{figure}[!ht]
    \centering
    \includegraphics[width=0.76\columnwidth]{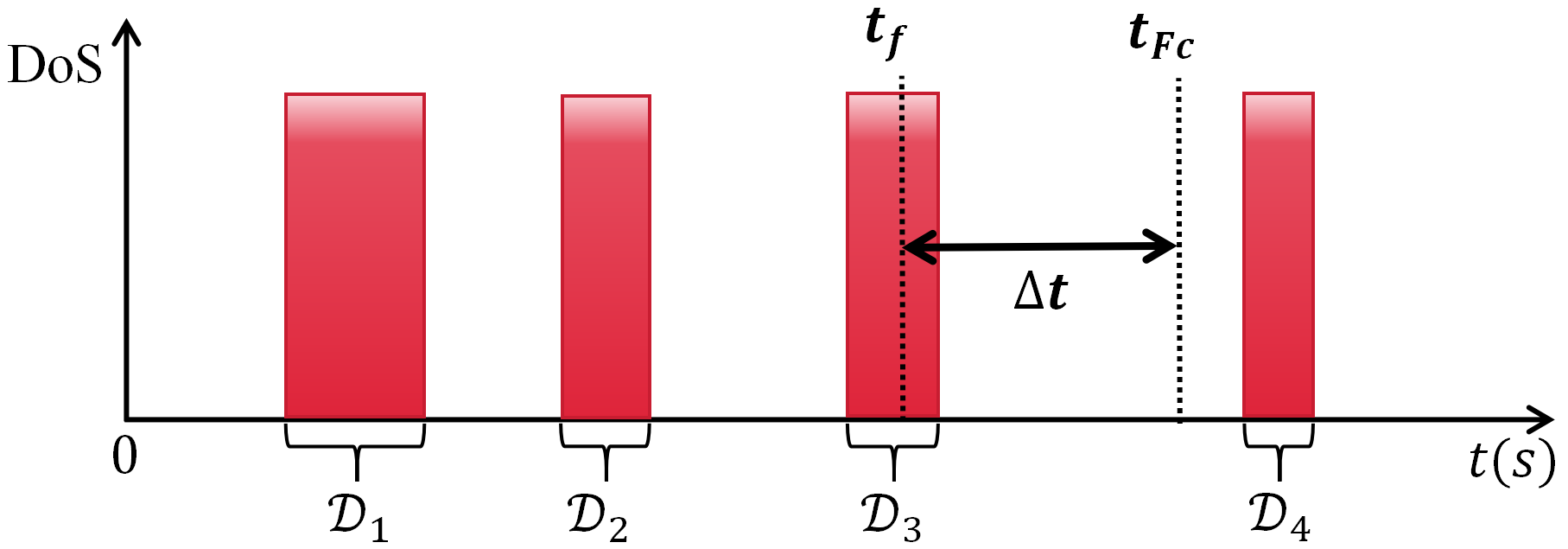}
    \caption{Prescribed-time extension under DoS attacks, where $\mathcal{D}_l$ denotes the
$l$th DoS attack interval. Here, $t_f$ is the prescribed time without DoS
attacks, while $t_{Fc}$ is the extended prescribed time under DoS attacks. The
minimum required extension $\Delta t$ is given in \eqref{eq:Delta_t_explicit}.}
    \label{fig:DoS_time}
\end{figure}
To overcome this limitation, an extended prescribed time $t_{Fc} = t_f + T_d(0,t_{Fc})$ is introduced to replace $t_f$. Under Assumption \ref{ass_dos}, the finite number of DoS transitions guarantees the existence of at least one $t_{Fc}$ residing in a DoS-free interval. However, this construction remains implicit because the exact cumulative attack duration $T_d(0,t_{Fc})$. i.e., $|\mathcal{D}_1|+|\mathcal{D}_2|+|\mathcal{D}_3|$ in Fig. \ref{fig:DoS_time}, is generally unknown a priori. Furthermore, the framework in \cite{ref-Fenggang-Auto} is restricted to linear systems with full-state measurements, leaving nonlinear follower dynamics and output-feedback implementations unaddressed. In this section, while $t_{Fc}$ is similarly adopted as the extended prescribed time, the required time extension $\Delta t = t_{Fc}-t_f$ can be explicitly determined. Accordingly, an explicit upper bound on the allowable DoS attacks is established, and the time-varying scaling function is redefined over the extended interval as
\begin{equation}
\alpha(t) = \frac{1}{t_{Fc} - t}, \quad t \in [0, t_{Fc}).
\label{eq:alpha_tfc}
\end{equation}

\begin{theorem}\label{Theorem 2}
	\textit{(A separation principle for consensus tracking under DoS attacks)}
	Consider the MAS consisting of a leader~\eqref{eq:leader} and \(N\) followers~\eqref{eq:system} under Assumptions~\ref{ass:nonlinearity}--\ref{ass_dos}. Suppose that the hybrid observer, composed of the switched distributed observer~\eqref{eq:switched_observer} and the local observer~\eqref{eq:local observer}, together with the event-triggered controller~\eqref{ui_ETC} under the triggering condition~\eqref{triggering condition}, is applied. Then,  the following results hold: $i)$ the  observer error convergences to zero within the prescribed time $t_{F_c}$ and the prescribed-time consensus tracking is achieved in the sense of Definition~1; $ii)$ Zeno behavior is excluded; $iii) $ there exist positive constants \(\hat{m}_i\), \(\hat{M}_i\), \(\hat{m}_{u_i}\), and \(\hat{M}_{u_i}\) such that, for \(\chi_i := [\tilde{x}_i^\top, (e_i+\zeta_i)^\top]^\top\),
		\begin{align}
			|\chi_i(t)| &\le (t_{Fc}-t)^{\hat{m}_i}\hat{M}_i|\chi_i(0)|, \nonumber\\
			|u_i(t)| &\le (t_{Fc}-t)^{\hat{m}_{u_i}}\hat{M}_{u_i}|\chi_i(0)|; \nonumber
		\end{align}     and $iv)$ the control scheme is resilient to DoS attacks provided that the maximum attack duration satisfies
		\begin{equation}\label{constraint DoS duration}
			\frac{1}{\tau_d} + \frac{k_{dT}}{t_{Fc}} < \frac{\Upsilon_i}{\lambda_{Di}\gamma_{oi}},
		\end{equation}
		where \(\lambda_{Di} > 0\) is a constant  and \(\Upsilon_i\) is defined in~\eqref{dot V con reduce}.
	Furthermore, the design of the output-ffedback controller and the hybrid observer satisfies the separation principle in the sense that the coefficients \(\gamma_{ci}\) and \(\gamma_{oi}\) associated with the controller and observer gains can be independently selected from  the sets defined in~\eqref{Omega_ci oi}.
\end{theorem}

\begin{proof}
The proof is divided into two parts.

    \textbf{Part 1.} 
Similarly, the Lyapunov candidate function under DoS attacks is designed as $V_{Di} = V_{ci} + c_i V_{oi}$, where $V_{ci}$ and $V_{oi}$ are the same as those in Section IV. 

During DoS attacks ($t \in \Phi_D$), the distributed observer switches to the autonomous prediction mode, yielding \eqref{new zeta i}. Consequently, taking the time derivative of $c_i V_{oi}$ introduces the coupled term
\begin{align}
    2 c_i\varrho_i^{2n+1} \zeta_i^{\top} \mathbf{Q}_i^{-1} \dot{\zeta}_i = c_i\varrho_i^{2n+1} \zeta_i^{\top} (\mathbf{Q}_i^{-1}A + A^{\top}\mathbf{Q}_i^{-1}) \zeta_i.
\end{align}
Substituting \eqref{eq:PLE2} into the derivative of $c_i V_{oi}$ yields
\begin{equation}
2 c_i\varrho_i^{2n+1} \zeta_i^{\top} \mathbf{Q}_i^{-1} \dot{\zeta}_i = c_i\varrho_i^{2n+1} \zeta_i^{\top} C^{\top}C \zeta_i - c_i \varrho_i^{2n+2} \zeta_i^{\top} \mathbf{Q}_i^{-1} \zeta_i.
\nonumber
\end{equation}
Note that the second term on the right-hand side corresponds exactly to the term derived in the attack-free scenario, where $\dot{\zeta}_i = -\frac{\varrho_i}{2}\zeta_i$. Therefore, following the derivation process detailed in Section IV, the upper bound of $\dot{V}_{Di}$ incurs only an additional positive term, $c_i\varrho_i^{2n+1}\zeta_i^{\top} C^{\top}C \zeta_i$, which is induced by the DoS attacks. Thus, the derivative of $V_{Di}$ yields
\begin{equation}
\dot{V}_{Di} \le -\Upsilon_i \alpha(t)V_{Di} + s_i V_{Di} + c_i\varrho_i^{2n+1}\zeta_i^{\top} C^{\top}C \zeta_i,
\label{eq:V_D_derivative step 1}
\end{equation}
where, as in Section~4,
\[
\begin{aligned}
\Upsilon_i &= \min\left\{
\kappa_{ci}\gamma_{ci},
\frac{\kappa_{oi}\gamma_{oi}}{c_i}
\right\},\\
s_i &= \max\left\{
\sigma(\mathbf{Q}_{i,s2},-\mathbf{Q}_{i,s1})
\right\}.
\end{aligned}
\]
To address the additional term $c_i\varrho_i^{2n+1}\zeta_i^{\top} C^{\top}C \zeta_i$ resulting from the DoS-induced switching, we employ the resilient matrix-pencil formulation. By virtue of Lemma \ref{lemma_matrix pencil}, we select a positive scalar $\lambda_{Di} = \max\{\sigma(C^{\top}C,\mathbf{Q}_{ni}^{-1})\}$ such that the following generalized inequality $C^{\top}C \le \lambda_{Di} \mathbf{Q}_{ni}^{-1}$ holds.
Based on \eqref{Lic Lio} and the definition of $C$, we have $\mathbf{L}_{oi} C^{\top}C \mathbf{L}_{oi} = \varrho_i^{2n-2} C^{\top}C$. Then, applying \eqref{eq:7} yields $C^{\top}C \le \lambda_{Di} \varrho_i \mathbf{Q}_i^{-1}$, which implies
\begin{align}
    c_i\varrho_i^{2n+1}\zeta_i^{\top} C^{\top}C \zeta_i \le \lambda_{Di} \varrho_i \left( c_i \varrho_i^{2n+1} \zeta_i^{\top} \mathbf{Q}_i^{-1} \zeta_i \right) \le \lambda_{Di} \varrho_i c_i V_{oi}. \nonumber
\end{align}
Consequently, we obtain
\begin{equation}
\dot{V}_{Di} \le (-\Upsilon_i+ \lambda_{Di}\gamma_{oi}) \alpha(t)V_{Di}  + s_i V_{Di}.
\label{eq:V_D_derivative step 2}
\end{equation}
Combining with the derivative of the Lyapunov function during 
$t \in \Phi_{\mathcal{H}}$ in \eqref{dot V con reduce}, the piecewise 
differential inequality for the global Lyapunov function over $t \ge 0$ 
can be formulated as
\begin{align}
\dot{V}_i
\le
\varpi_i(t)\alpha(t)V_i(t)
+
s_iV_i(t),
\label{eq:piecewise_Vi}
\end{align}
where the switching coefficient $\varpi_i(t)$ is defined as
\begin{align}
\varpi_i(t)
=
\begin{cases}
-\Upsilon_i, 
& t \in \Phi_{\mathcal{H}},\\
\lambda_{Di}\gamma_{oi}-\Upsilon_i, 
& t \in \Phi_{\mathcal{D}}.
\end{cases}
\label{eq:mu_i_def}
\end{align}
Integrating the piecewise dynamics over the interval $[0, t]$ for any $t \in [0, t_{Fc})$, we obtain
\begin{equation}
V_i(t) \le V_i(0) e^{s_i t_{Fc}} \exp\left( \int_0^t \varpi_i(\tau)\alpha(\tau) d\tau \right).
\end{equation}
Notice that the integral of the switching coefficient relative to the scaling function can be expanded as
\begin{equation}
\int_0^t \varpi_i(\tau)\alpha(\tau)d\tau = -\Upsilon_i \int_0^t \alpha(\tau)d\tau + \lambda_{Di}\gamma_{oi} \int_{\Phi_{\mathcal{D}}(0,t)} \alpha(\tau)d\tau.\nonumber
\end{equation}
Define an $\alpha$-weighted DoS ratio over $[0, t]$ as $\delta_i^\alpha(t) := {\int_{\Phi_{\mathcal{D}}(0,t)} \alpha(\tau) d\tau}/{\int_0^t \alpha(\tau) d\tau}$, allowing us to rewrite the above inequality as
\begin{equation}
\int_0^t \varpi_i(\tau)\alpha(\tau)d\tau = \big(-\Upsilon_i + \lambda_{Di}\gamma_{oi} \delta_i^\alpha(t)\big) \int_0^t \alpha(\tau)d\tau.
\label{eq:integral_with_delta_alpha}
\end{equation}
Based on the definition of $\alpha(t)$, the $\alpha$-weighted DoS ratio is bounded by a maximum physical DoS duty cycle $\delta_D(t_{Fc}):={T_d(0,t_{Fc})}/{t_{Fc}}$, i.e., $\delta_i^\alpha(t) \le \delta_D(t_{Fc})$. Substituting this inequality into \eqref{eq:integral_with_delta_alpha} yields:
\begin{equation}
\int_0^t \varpi_i(\tau)\alpha(\tau)d\tau \le \big(-\Upsilon_i + \lambda_{Di}\gamma_{oi} \delta_D(t_{Fc})\big) \int_0^t \alpha(\tau)d\tau.
\end{equation}
According to Assumption \ref{ass_dos}, it holds that $\delta_D(t_{Fc})\le \frac{1}{\tau_d} + \frac{k_{dT}}{t_{Fc}}$. Combining with the constraint \eqref{constraint DoS duration}, we can obtain a positive constant $\hat{\Upsilon}_{i} = \Upsilon_i - \lambda_{Di}\gamma_{oi}\big(\frac{1}{\tau_d} + \frac{k_{dT}}{t_{Fc}}\big) > 0$. 
Based on Lemma 1 in \cite{ref4}, it yields
\begin{equation}\label{Vi  DOS upper}
V_i(t) \le \left( \frac{t_{Fc}-t}{t_{Fc}} \right)^{\hat{\Upsilon}_{i}} e^{s_i t} V_i(0),
\qquad t\in[0,t_{Fc}).
\end{equation}
Following the same procedure as the proof of Theorem 1, the positive constants $\hat{m}_i = \frac{\hat{\kappa}_{1i}}{2} + 1$, $\hat{M}_i$, $\hat{m}_{u_i}$, and $\hat{M}_{u_i}$ can be successfully established, guaranteeing prescribed-time consensus tracking and all closed-loop signals are uniformly bounded.

\textbf{Part 2.} Since $V_i(t)$ remains bounded as shown in \eqref{Vi DOS upper}, and the event-triggering mechanism of the controllers is unaffected by DoS attacks, the exclusion of Zeno behavior can be proven using the exact same rationale as in Theorem 1. Therefore, the detailed proof is omitted for brevity.
\end{proof}
\begin{remark}
The matrix-pencil technique plays a key role in deriving the admissible DoS duty-cycle constraint. Specifically, it provides explicit generalized-eigenvalue conditions for dominating the coupling terms induced by DoS attacks and provides a constant $\lambda_{Di}$ to establish an explicit DoS duration constraint \eqref{constraint DoS duration}. 
\end{remark}
\begin{remark}
     Based on the resilience constraint in \eqref{constraint DoS duration}, one can deduce the lower bound for the extended time as 
     \begin{align}
         t_{Fc} > k_{dT} \big/ \big(\frac{\Upsilon_i}{\lambda_{Di}\gamma_{oi}} - \frac{1}{\tau_d}\big).
     \end{align} 
     Provided in the proof of Theorem 2 that the necessary condition ${1}/{\tau_d} < {\Upsilon_i}/{\lambda_{Di}\gamma_{oi}}$ is satisfied, the extension interval $\Delta t$ can be explicitly designed as
\begin{equation}
\Delta t > \frac{k_{dT} \lambda_{Di} \gamma_{oi} \tau_d}{\Upsilon_i \tau_d - \lambda_{Di}\gamma_{oi}} - t_f.
\label{eq:Delta_t_explicit}
\end{equation}
Note that this constraint can be satisfied by properly adjusting the parameters
of the hybrid observer, thereby guaranteeing the realization of the separation
principle. 
\end{remark}

\section{Numerical examples}
\begin{figure}[!h]
    \centering
    \includegraphics[width=0.8\columnwidth]{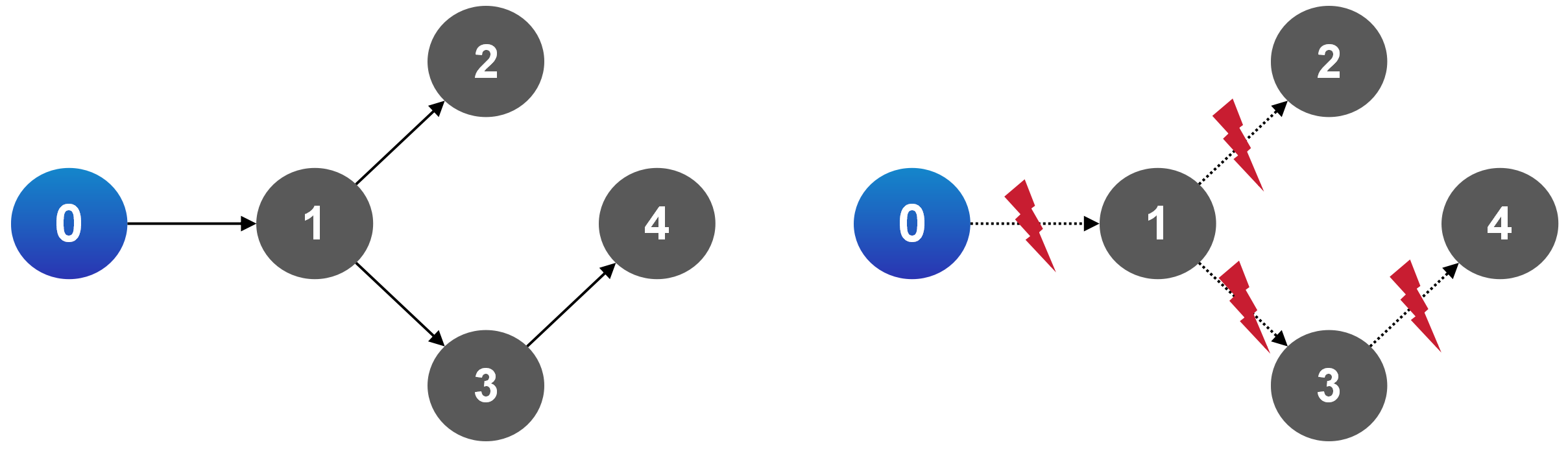}
    \caption{The communication topology without and under DoS attacks}
    \label{fig:digraph}
\end{figure}
We consider an MAS with one leader and four followers, whose communication network is described by a directed graph shown as Fig. \ref{fig:digraph}. The dynamic of the $i$th agent is described by a single-link robotic manipulator model as follows (\cite{ref-ye2025distributed}):
\begin{equation}\label{simu model}
J_i \ddot{q}_i + K_i \dot{q}_i + m_i g d_i \sin(q_i) = u_i,
\quad i = 0,1,\ldots,4.
\end{equation}
 Here, $q_i$ and $\dot{q}_i$ denote the angular position and angular velocity, respectively. The parameter $J_i$ represents the moment of inertia, $K_i$ is the damping coefficient, $d_i$ is the distance from the joint axis to the center of mass, $m_i$ denotes the mass of the link, and $g$ is the gravitational acceleration. In the simulation, the parameters are selected as $J_0 = J_4 = 8.5$, $J_1 = J_2 = J_3 = 10$, $K_0 = K_3 = 1.4$, $K_1 = K_2 = K_4 = 1.6$, $d_0 = d_1 = 1$, $d_2 = 0.8$, $d_3 = d_4 = 1.2$, $m_0 = m_2 = 1.3$, $m_1 = m_3 = m_4 = 1$, and $g = 9.8$. By defining $x_i = [x_{i,1},x_{i,2}]^\top = [q_i,\dot{q}_i]^\top$ and $y_i = q_i$, the dynamic in \eqref{simu model} can be recast into the form of \eqref{eq:system}, where $\mathbf{A} = [0,1;0,0]$, $B = [0;1]$, $C = [1;0]$, and $f_i = \bigl[0,-(K_i/J_i)\dot{q}_i-(m_i g d_i/J_i)\sin(q_i)\bigr]^\top$. For the leader, $u_0 = 0$, and $f_0 = (\cdot,t)$ is available for the overall MAS. For the followers, namely $i=1,\ldots,4$, the nonlinear terms $f_i$ are unknown functions. The initial states are set as: $x_0 = [0,0]^\top$, $x_1 = [1,1]^\top$, $x_2 = [2,2]^\top$, $x_3 = [3,3]^\top$ and $x_4 = [4,4]^\top$.

Two situations are considered in this section to validate the effectiveness of Theorem 1 and Theorem 2, respectively.

\subsection{Consensus tracking without DoS attacks}
According to Table \ref{tab:parameter_ranges}, we select $\gamma_{oi} = 10$ and $\gamma_{ci} = 15$ for all followers in this case. The prescribed time is designed as $t_f = 3s$.

\begin{figure}[!htbp]
	\centering
	{
		\includegraphics[width=0.82\columnwidth]{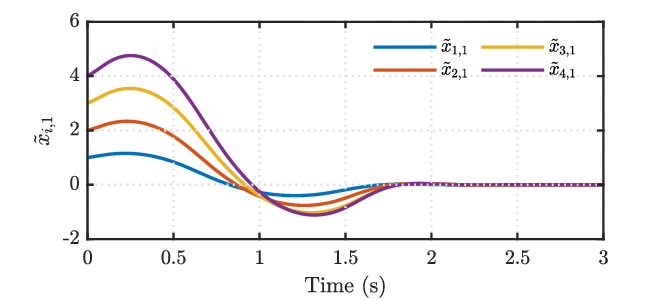}
		\label{fig:tracking_1_no_dos}
	}
	
	\vspace{-1.2em}
	{
		\includegraphics[width=0.82\columnwidth]{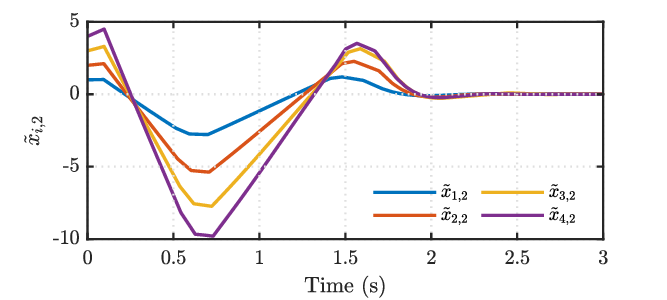}
		\label{fig:tracking_2_no_dos}
	}
	\caption{Tracking errors $\tilde{x}_i$ in the absence of DoS attacks.}
	\label{fig:no_dos_tracking}
\end{figure}

\begin{figure}[!htbp]
	\centering
	{
		\includegraphics[width=0.82\columnwidth]{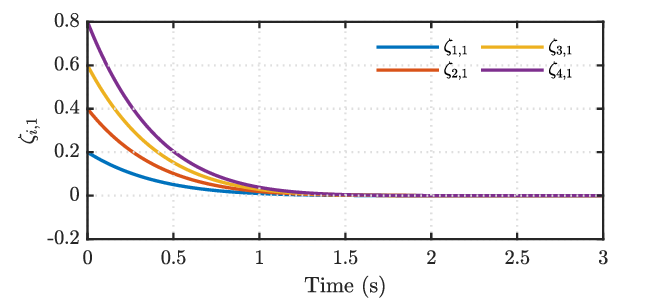}
		\label{fig:zeta_1_no_dos}
	}
	
	\vspace{-1.2em}
	{
		\includegraphics[width=0.82\columnwidth]{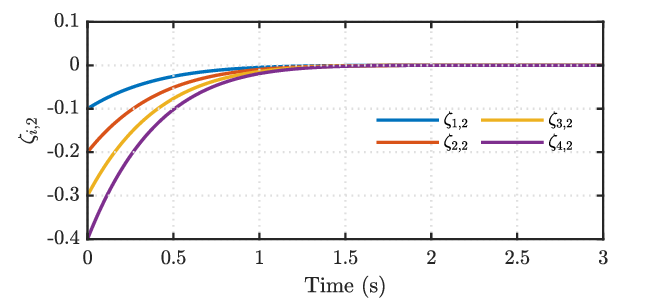}
		\label{fig:zeta_2_no_dos}
	}
	\caption{Distributed observation errors $\|\zeta_i\|$ in the absence of DoS attacks.}
	\label{fig:no_dos_distributed}
\end{figure}

\begin{figure}[!htbp]
	\centering
	{
		\includegraphics[width=0.82\columnwidth]{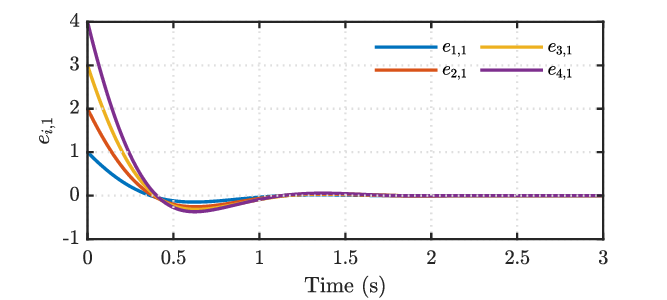}
		\label{fig:e_1_no_dos}
	}
	
	\vspace{-1.2em}
	{
		\includegraphics[width=0.82\columnwidth]{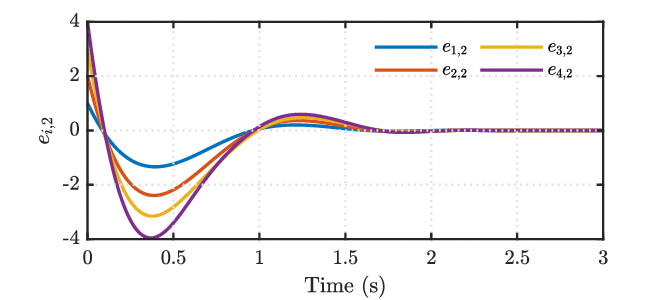}
		\label{fig:e_2_no_dos}
	}
	\caption{Local observation errors $e_i$ in the absence of DoS attacks.}
	\label{fig:no_dos_local}
\end{figure}

\begin{figure}[!htbp]
	\centering
	\includegraphics[width=0.82\columnwidth]{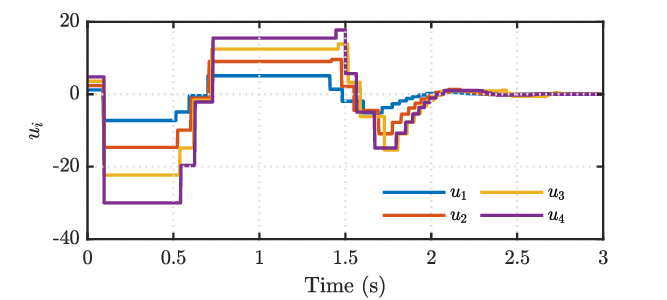}
	\caption{Control inputs $u_i$ in the absence of DoS attacks.}
	\label{fig:no_dos_control}
\end{figure}

As shown in Fig.~\ref{fig:no_dos_tracking}--Fig.~\ref{fig:no_dos_local}, the tracking errors 
$\tilde x_i$, the distributed observation errors $\zeta_i$, and the local 
observation errors $e_i$ all approach zero before the prescribed time 
$t_f=3$s. Consequently, the leader-follower consensus tracking objective 
in the absence of DoS attacks is achieved. As shown in Fig.~\ref{fig:no_dos_control}, the control inputs are piecewise constant due to the event-triggered mechanism. Meanwhile, all 
closed-loop signals remain ultimately bounded, which is consistent with Theorem~\ref{Theorem 1}. Moreover, the exclusion of Zeno behavior is demonstrated in Table \ref{tab:sampling_times_comparison}, where the simulation step size is $dt=1.0\times10^{-4}$ s.

\begin{table}[!t]
	\centering
	\caption{Comparison of Sampling Times and Minimum Inter-Event Times}
	\label{tab:sampling_times_comparison}
	\renewcommand{\arraystretch}{1.08}
	\setlength{\tabcolsep}{1.7pt}
	\scriptsize
	\begin{tabular}{cccccc}
		\toprule
		Scenario & Type & $u_1$ & $u_2$ & $u_3$ & $u_4$ \\
		\midrule
		\multirow{3}{*}{No DoS}
		& TT 
		& $30000$ & $30000$ & $30000$ & $30000$ \\
		& ET 
		& $\mathbf{97}$ & $\mathbf{103}$ & $\mathbf{94}$ & $\mathbf{102}$ \\
		& \makecell{Min.\\IET}
		& $6.0{\times}10^{-4}$ 
		& $5.0{\times}10^{-4}$ 
		& $6.0{\times}10^{-4}$ 
		& $5.0{\times}10^{-4}$ \\
		\midrule
		\multirow{3}{*}{DoS}
		& TT 
		& $41500$ & $41500$ & $41500$ & $41500$ \\
		& ET 
		& $\mathbf{98}$ & $\mathbf{68}$ & $\mathbf{91}$ & $\mathbf{78}$ \\
		& \makecell{Min.\\IET}
		& $7.0{\times}10^{-4}$ 
		& $6.0{\times}10^{-4}$ 
		& $5.0{\times}10^{-4}$ 
		& $6.0{\times}10^{-4}$ \\
		\bottomrule
	\end{tabular}
	\vspace{1mm}
	
	\raggedright
	\scriptsize
	TT: time-triggering; ET: event-triggering; IET: inter-event time.
\end{table}
\subsection{Consensus tracking under DoS attacks}

Fig.~\ref{fig:dos_tracking}--Fig. \ref{fig:dos_control} present the simulation results under DoS attacks.
The shaded regions of all figures denote the DoS-active intervals, during which all 
communication links are unavailable, and the distributed observer switches from 
the communication-based mode to the autonomous prediction mode in 
\eqref{eq:switched_observer}. In this simulation, the parameters are selected as $\gamma_{oi} = 10$ and $\gamma_{ci} = 15$ for all followers. The normal prescribed time $t_f$ is $3$s. The extended prescribed time is designed as $t_{Fc} = 4.15s$, and the total duration of DoS attacks is selected as $T_D(0,t_{Fc}) = 1.25s$, under which the DoS duty cycle constraint \eqref{constraint DoS duration} is satisfied, since $\lambda_{Di} = 2.0$ and $\Upsilon_i = 4.2$.

\begin{figure}[!htbp]
	\centering
	{
		\includegraphics[width=0.82\columnwidth]{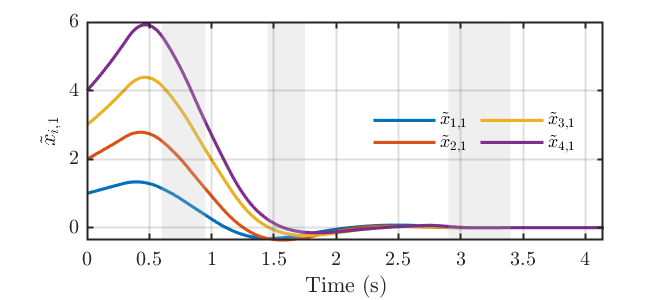}
		\label{fig:tracking_1_dos}
	}
	
	\vspace{-1.2em}
	{
		\includegraphics[width=0.82\columnwidth]{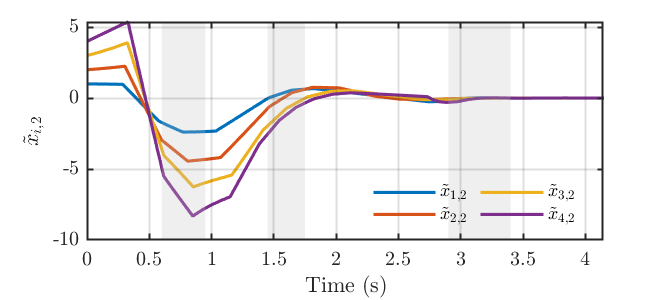}
		\label{fig:tracking_2_dos}
	}
	\caption{Tracking errors $\tilde{x}_i$ under DoS attacks.}
	\label{fig:dos_tracking}
\end{figure}

\begin{figure}[!htbp]
	\centering
	{
		\includegraphics[width=0.82\columnwidth]{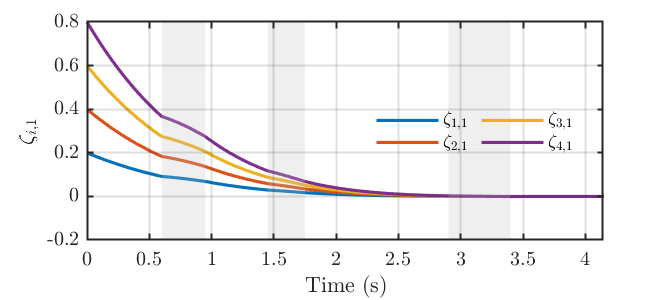}
		\label{fig:zeta_1_dos}
	}
	
	\vspace{-1.2em}
	{
		\includegraphics[width=0.82\columnwidth]{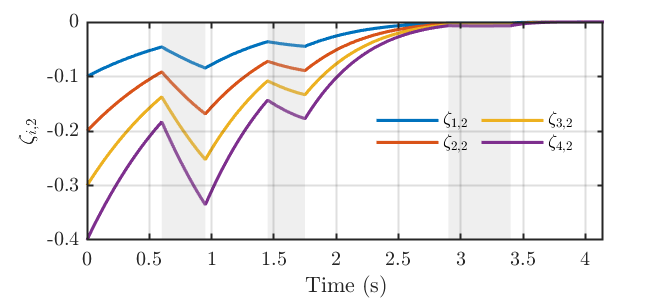}
		\label{fig:zeta_2_dos}
	}
	\caption{Distributed observation errors $\zeta_i$ under DoS attacks.}
	\label{fig:dos_distributed}
\end{figure}

\begin{figure}[!htbp]
	\centering
	{
		\includegraphics[width=0.82\columnwidth]{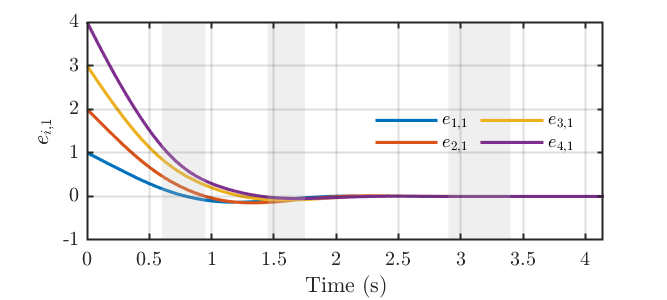}
		\label{fig:e_1_dos}
	}
	
	\vspace{-1.2em}
	{
		\includegraphics[width=0.82\columnwidth]{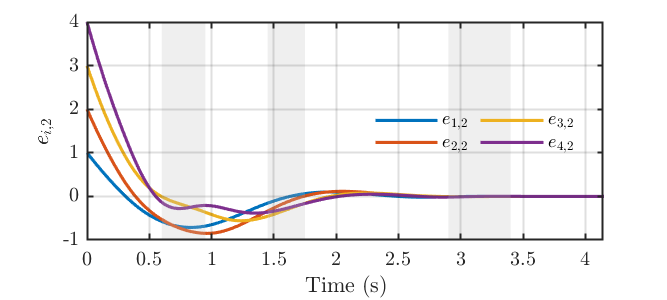}
		\label{fig:e_2_dos}
	}
	\caption{Local observation errors $e_i$ under DoS attacks.}
	\label{fig:dos_local}
\end{figure}
\begin{figure}[!htbp]
	\centering
	\includegraphics[width=0.82\columnwidth]{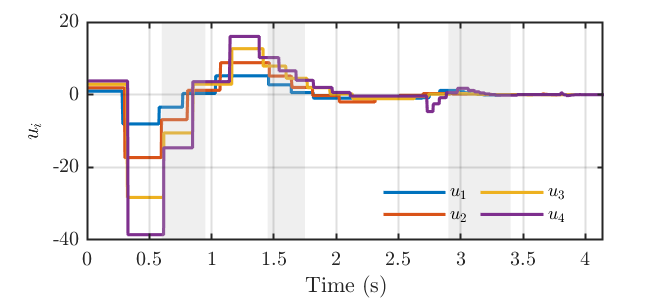}
	\caption{Control inputs $u_i$ under DoS attacks.}
	\label{fig:dos_control}
\end{figure}

As shown in Fig.~\ref{fig:dos_tracking}--Fig. \ref{fig:dos_local}, the tracking errors 
$\tilde x_i$, the distributed observation errors $\zeta_i$, and the local 
observation errors $e_i$ all approach zero before the extended prescribed 
time $t_{Fc}$. Fig.~\ref{fig:dos_control} shows the corresponding control inputs. These results are consistent with Theorem~\ref{Theorem 2} and demonstrate the resilience of the 
proposed scheme against intermittent communication denial. Similarly, the exclusion of Zeno behavior is demonstrated in Table \ref{tab:sampling_times_comparison}, where the simulation step size is $dt=1.0\times10^{-4}$ s.

To further demonstrate the advantages of the proposed method, a comparison is conducted under the same simulation setting. For a fair implementation, the method in \cite{ref-Fenggang-Auto} is implemented using full-state information and without event-triggered control, and the gain of the time-varying function is selected as $\alpha=3>0$, as required in \cite{ref-Fenggang-Auto}. The comparison results are presented in Fig. \ref{fig:compare}. It can be observed that the tracking errors under the method in \cite{ref-Fenggang-Auto} are not effectively suppressed in the presence of system nonlinearities and the directed communication topology, whereas the proposed method achieves accurate prescribed-time consensus tracking under DoS attacks.
\begin{figure}[!htbp]
	\centering
	{
		\includegraphics[width=0.81\columnwidth]{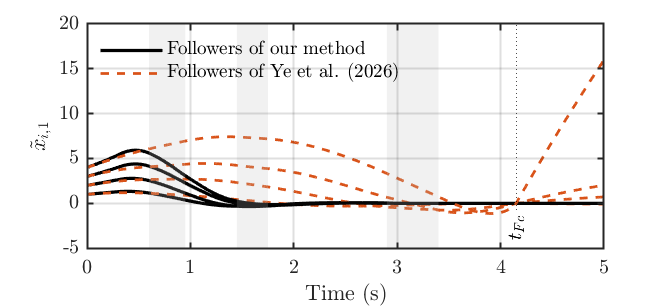}
		\label{Compare 1}
	}
	
	\vspace{-1.2em}
	{
		\includegraphics[width=0.81\columnwidth]{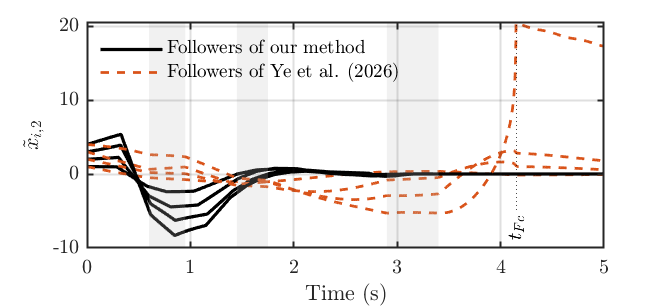}
		\label{Compare 2}
	}
	\caption{Comparison of tracking errors between the proposed method and the method in \cite{ref-Fenggang-Auto} under DoS attacks.}
	\label{fig:compare}
\end{figure}

\section{Conclusion}
This paper investigated prescribed-time leader-following consensus tracking for
high-order nonlinear MASs under DoS attacks with output-only measurements. A
distributed event-triggered output-feedback framework is developed by combining
a local observer, a switched distributed observer, and an
event-triggered controller. Based on PLEs and a resilient matrix pencil
analysis, a separation principle was established, and explicit design conditions are derived to handle nonlinearities, observer errors, tracking errors, and DoS-induced switching terms. The closed-loop MAS is shown to
achieve consensus tracking within the prescribed time without DoS attacks and
within an extended prescribed time under admissible DoS attacks. Zeno behavior
is excluded. Future work will consider heterogeneous MASs, relaxed attack
models, and experimental validation.

\bibliographystyle{cas-model2-names}
\bibliography{refs}

{\small

\Needspace{10\baselineskip}
{
\setlength{\columnsep}{6pt} % distance between photo and text
\begin{wrapfigure}[9]{l}{1.02in}
    \vspace{-6pt}
    \centering
    \includegraphics[width=1.00in,height=1.25in,clip,keepaspectratio]{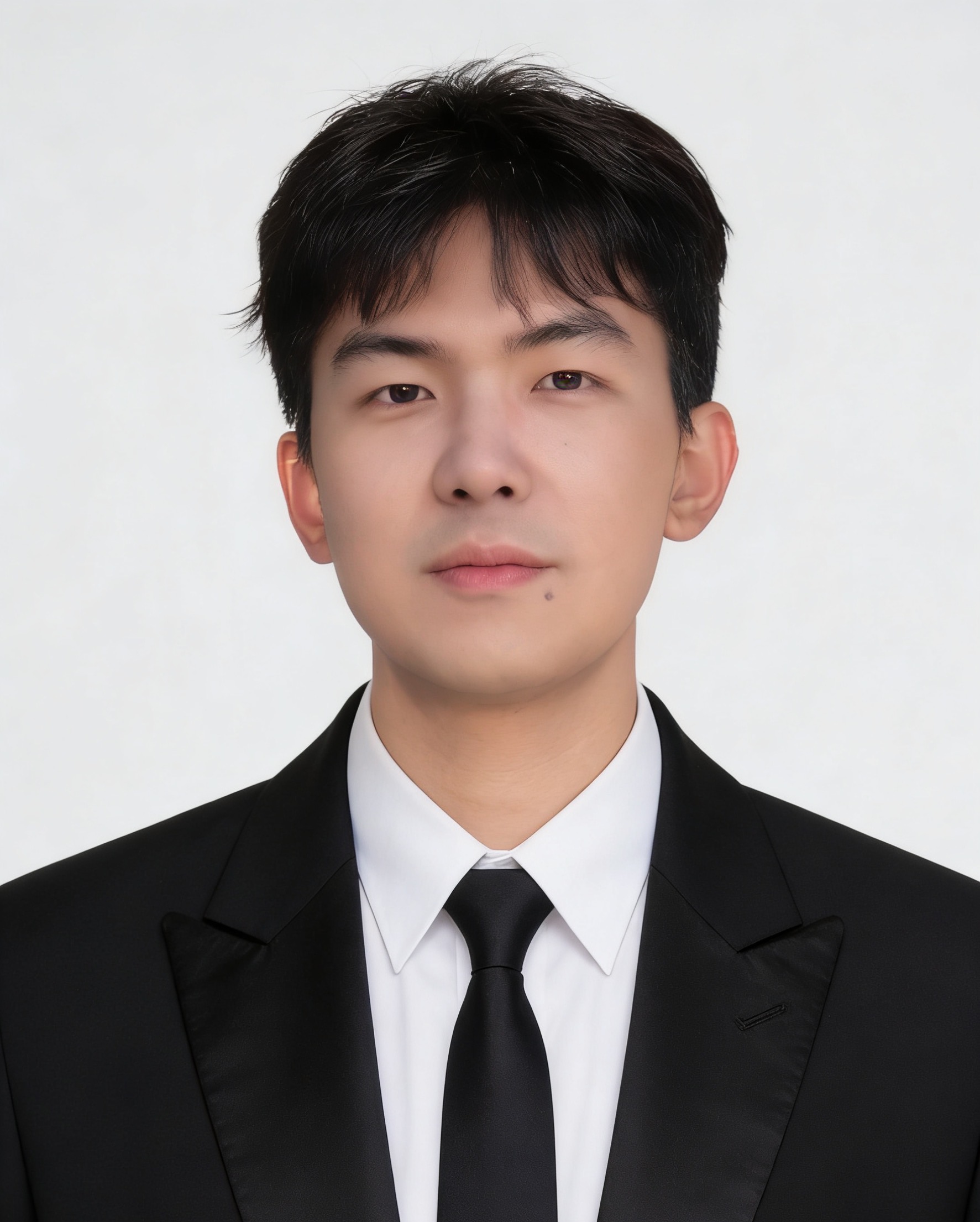}
    \vspace{-10pt}
\end{wrapfigure}
\noindent\textbf{Hongjian Chen} received the B.Eng. degree in Electronic Information Engineering from the University of Electronic Science and Technology of China, Chengdu, China, and the M.Sc. degree in Computer Control and Automation from Nanyang Technological University, Singapore, where he is currently pursuing the Ph.D. degree with the School of Electrical and Electronic Engineering. His current research interests include resilient control, complex systems, and multi-agent systems.
\par\vspace{0.3em}
}

\Needspace{10\baselineskip}
{
\setlength{\columnsep}{6pt} % distance between photo and text
\begin{wrapfigure}[10]{l}{1.02in}
    \vspace{-6pt}
    \centering
    \includegraphics[width=1.00in,height=1.25in,clip,keepaspectratio]{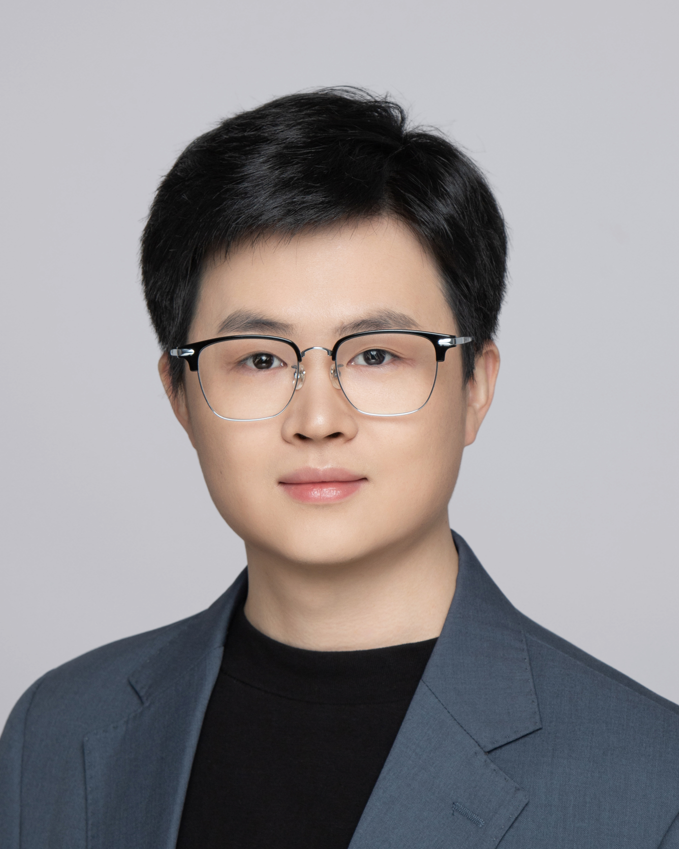}
    \vspace{-6pt}
\end{wrapfigure}
\noindent\textbf{Hefu Ye} received the B.Eng. degree from Harbin Institute of Technology in 2019 and the Ph.D. degree from Chongqing University in 2025. During 2022 and 2023, he was a Joint Ph.D. student with the School of Electrical and Electronic Engineering, Nanyang Technological University, Singapore. Following the completion of his Ph.D., he served as a Research Associate with The University of Hong Kong. Since September 2025, he has been a Postdoctoral Fellow with the University of Macau. His current research interests include prescribed-time control, robotics, and learning for control applications.
\par\vspace{0.3em}
}

% Dr. Ye is the recipient of the 2025 Andrew P. Sage Best Transactions Paper Award from the IEEE Systems, Man, and Cybernetics Society, as well as the Best Theory Paper Award at the China Automation Congress 2025.
\par\vspace{0.5em}

\Needspace{10\baselineskip}
{
\setlength{\columnsep}{6pt} % distance between photo and text
\begin{wrapfigure}[10]{l}{1.02in}
    \vspace{-6pt}
    \centering
    \includegraphics[width=1.00in,height=1.25in,clip,keepaspectratio]{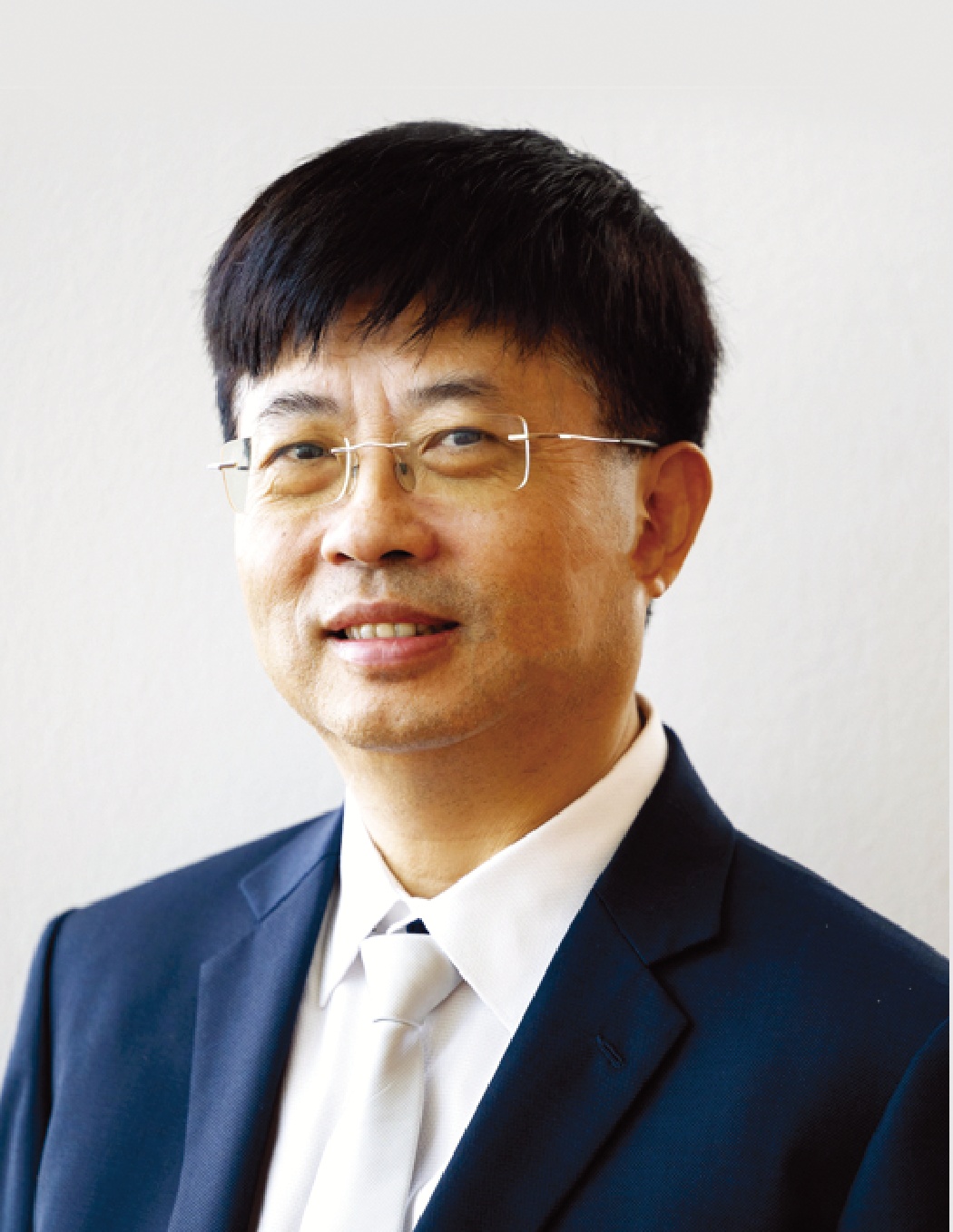}
    \vspace{-10pt}
\end{wrapfigure}
\noindent\textbf{Changyun Wen} received the B.Eng. degree from Xi'an Jiaotong University in 1983, and the Ph.D. degree from the University of Newcastle, Australia in 1990. From August 1989 to August 1991, he was a Postdoctoral Fellow at the University of Adelaide, Australia. Since August 1991, he has been with Nanyang Technological University, Singapore, where he is currently a Full Professor. Prof. Wen is a Fellow of IEEE and Fellow of the Academy of Engineering, Singapore. His main research activities are in the areas of control systems and applications, cyber-physical systems, smart grids, complex systems and networks.
\par\vspace{0.3em}
}

% Prof. Wen is a Fellow of IEEE and Fellow of the Academy of Engineering, Singapore. He was a member of IEEE Fellow Committee from January 2011 to December 2013 and a Distinguished Lecturer of IEEE Control Systems Society from 2010 to 2013. He is currently the co-Editor-in-Chief of IEEE Transactions on Industrial Electronics, Associate Editor of Automatica, and Executive Editor-in-Chief of Journal of Control and Decision. He also served as an Associate Editor of IEEE Transactions on Automatic Control from 2000 to 2002, IEEE Transactions on Industrial Electronics from 2013 to 2020, and IEEE Control Systems Magazine from 2009 to 2019.
\par\vspace{0.5em}

}

\end{document}